\documentclass[10pt]{article}

\usepackage{xr-hyper}
\usepackage[colorlinks,citecolor=blue]{hyperref}
\usepackage{amsfonts,amsmath,amssymb,amsthm,booktabs}
\usepackage[pdftex]{graphicx}
\usepackage[hmargin=1in,vmargin=1in]{geometry}
\usepackage{natbib,setspace,enumerate,mathtools}
\usepackage{courier}
\usepackage{changepage}
\allowdisplaybreaks
\usepackage[toc,title,titletoc,header]{appendix}
\usepackage{etoolbox} 
\def\qed{\rule{2mm}{2mm}}
\def\independent{\perp \!\!\! \perp}

\newtheorem{theorem}{Theorem}[section]
\newtheorem{lemma}{Lemma}[section]

\theoremstyle{definition}

\newtheorem{remark}{Remark}[section]
\newtheorem{assumption}{Assumption}[section]

\AtEndEnvironment{remark}{~\qed}
\AtEndEnvironment{example}{~\qed}

\DeclareMathOperator*{\var}{Var}
\DeclareMathOperator*{\cov}{Cov}

\newtheorem{myexa}{Example}[section]

\newenvironment{myexacont}
{\addtocounter{myexa}{-1}\begin{myexa}{\hspace{-1mm}\textbf{{(cont.)}}}}
  {\end{myexa}}

\begin{document}

\author{
Federico Bugni \\
Department of Economics\\
Northwestern University \\
\url{federico.bugni@northwestern.edu}
\and
Ivan A.\ Canay \\
Department of Economics\\
Northwestern University\\
\url{iacanay@northwestern.edu}
\and
Azeem M.\ Shaikh\\
Department of Economics\\
University of Chicago \\
\url{amshaikh@uchicago.edu}
\and
Max Tabord-Meehan\\
Department of Economics\\
University of Chicago \\
\url{maxtm@uchicago.edu}
}

\bigskip

\title{Inference for Cluster Randomized Experiments \\ with Non-ignorable Cluster Sizes\thanks{We thank the Coeditor and three anonymous referees for comments and suggestions that have improved the manuscript. We also would like to thank Eric Auerbach, David MacKinnon, Joe Romano and conference participants at Cowles, TSE, and ASSA 2023 for helpful comments on this paper. Xun Huang and Juri Trifonov provided excellent research assistance.}}

\maketitle

\vspace{-0.3in}

\begin{spacing}{1.2}
\begin{abstract}
This paper considers the problem of inference in cluster randomized experiments when cluster sizes are non-ignorable.  Here, by a cluster randomized experiment, we mean one in which treatment is assigned at the cluster level. By non-ignorable cluster sizes, we refer to the possibility that the individual-level average treatment effects may depend non-trivially on the cluster sizes. We frame our analysis in a super-population framework in which cluster sizes are random.  In this way, our analysis departs from earlier analyses of cluster randomized experiments in which cluster sizes are treated as non-random.  We distinguish between two different parameters of interest: the equally-weighted cluster-level average treatment effect, and the size-weighted cluster-level average treatment effect.  For each parameter, we provide methods for inference in an asymptotic framework where the number of clusters tends to infinity and treatment is assigned using a covariate-adaptive stratified randomization procedure.  We additionally permit the experimenter to sample only a subset of the units within each cluster rather than the entire cluster and demonstrate the implications of such sampling for some commonly used estimators.  A small simulation study and empirical demonstration show the practical relevance of our theoretical results.
\end{abstract}
\end{spacing}



\thispagestyle{empty} 
\newpage
\setcounter{page}{1}

\section{Introduction}
Cluster randomized experiments, in which treatment is assigned at the level of the cluster rather than at the level of the unit within a cluster, are widely used throughout economics and the social sciences more generally for the purpose of evaluating treatments or programs.  \cite{duflo2007using} survey various examples from development economics, in which clusters are villages and units within a cluster are households or individuals.  Numerous other examples can be found in, for instance, research on the effectiveness of educational interventions \citep[see, e.g.,][]{raudenbush1997statistical,schochet2013estimators,raudenbush2020randomized,schochet2021design} and research on the effectiveness of public health interventions \citep[see, e.g.,][]{turner2017review,donner2000design}. In this paper, we consider the problem of inference about the effect of a binary treatment on an outcome of interest in such experiments in a super-population framework in which cluster sizes are permitted to be random and non-ignorable. By non-ignorable cluster sizes, we refer to the possibility that the treatment effects may depend on the cluster size. 

Before proceeding, we illustrate this possibility with an example inspired by our empirical application in Section \ref{sec:application}. Suppose clusters represent public primary care facilities, or clinics, in a state. The size of the clinic may be related to the quality of care through a variety of mechanisms. For instance, it is plausible in our empirical application that patients may experience longer waiting times in smaller clinics because these may be comparatively understaffed relative to their patient populations. For this reason, patients may be less sensitive to incentives to go to the clinic more frequently. Similar considerations may, of course, apply in most other examples; we therefore view cluster sizes being non-ignorable as the rule rather than the exception.


To model this phenomenon, we adopt, in the spirit of the survey sampling literature \citep[see, e.g.,][]{lohr2021sampling}, a two-stage sampling framework, in which a set of clusters is first sampled from the population of clusters and then a set of units is sampled from the population of units within each cluster.  Importantly, in the first stage of the sampling process, each cluster may differ in terms of observed characteristics, including its size, and these characteristics may be used subsequently in the second stage of the sampling process to determine the number of units to sample from the cluster, including the possibility that all units in the cluster are sampled. We further emphasize that the sampling framework imposes no restrictions on the dependence across units within clusters. Our two-stage super-population sampling framework departs from earlier analyses of cluster randomized experiments in which cluster sizes were treated as non-random, and subsequently allows us to cohesively consider a large class of stratified treatment assignment mechanisms that can incorporate information on the cluster sizes along with other baseline covariates.  

In the context of this framework, we revisit two different parameters of interest previously considered in the literature on cluster randomized experiments \citep[see, for instance,][]{athey2017econometrics, su2021model,  wang2022two, kahan2023estimands} that differ in the way they aggregate, or average, the treatment effect across units. They differ, in particular, according to whether the units of interest are the clusters themselves or the individuals within the cluster. The first of these parameters takes the clusters themselves as the units of interest and identifies an \emph{equally-weighted} cluster-level average treatment effect. The second of these parameters takes the individuals within the clusters as the units of interest and identifies a \emph{size-weighted} cluster-level average treatment effect. When individual-level average treatment effects vary with cluster size (i.e., cluster size is non-ignorable) and cluster sizes are heterogeneous, these two parameters are generally different, though, as discussed in Remark \ref{rem:comparingestimands}, they coincide in some instances. Importantly, we show that the estimand associated with the standard difference-in-means estimator is a \emph{sample-weighted} cluster-level average treatment effect, which cannot generally be interpreted as an average treatment effect for either the clusters themselves or the individuals within the clusters. We show, however, in Section \ref{sec:parameters}, that this estimand can equal the \emph{size-weighted} or the \emph{equally-weighted} cluster-level average treatment effect for some very specific sampling designs. We argue that a clear description of whether the clusters themselves or the individuals within the clusters are of interest should therefore be at the forefront of empirical practice, yet we find that such a description is often absent. Indeed, we surveyed all articles involving a cluster randomized experiment published in the American Economic Journal: Applied Economics from $2018$ to $2022$. We document our findings in Appendix \ref{sec:lit_table}. 
From this survey, we find that most papers do not explicitly discuss their parameter of interest, and that as many as a third of the experiments conduct analyses that, when paired with their corresponding sampling design, do not necessarily recover either of the parameters that we consider in this paper. 

For each of the two parameters of interest we consider, we propose an estimator and develop the requisite distributional approximations to permit its use for inference about the parameter of interest when treatment is assigned using a covariate-adaptive stratified randomization procedure. In the case of the equally-weighted cluster-level average treatment effect, the estimator we propose takes the form of a difference-in-means of cluster averages.  This estimator may equivalently be described as the ordinary least squares estimator of the coefficient on treatment in a regression of the average outcome (within clusters) on a constant and treatment.  In the case of the size-weighted cluster-level average treatment effect, the estimator we propose takes the form of a \emph{weighted} difference-in-means of cluster averages, where the weights are proportional to cluster size.  This estimator may equivalently be described as the weighted least squares estimator of the coefficient on treatment in a regression of the individual-level outcomes on a constant and treatment with weights proportional to cluster size.\footnote{In Appendix \ref{sec:adjust}, we also briefly consider versions of both estimators which allow for linear regression adjustment using additional baseline covariates.}

Although both estimators we propose have previously been studied in the context of completely randomized experiments \citep[see][]{athey2017econometrics,su2021model}, to our knowledge we are the first to establish results for these estimators when treatment assignment is performed using a stratified covariate-adaptive randomization procedure. As in \cite{bugni/canay/shaikh:2018,bugni/canay/shaikh:2019}, this refers to randomization schemes that first stratify according to baseline covariates and then assign treatment status so as to achieve ``balance'' within each stratum (see \cite{rosenberger/lachin:2016} for a textbook treatment focused on clinical trials and \cite{duflo2007using} and \cite{bruhn/mckenzie:2008} for reviews focused on development economics.) Our results show that typical hypothesis tests constructed using a cluster-robust variance estimator are generally conservative in such cases, and as a result, we provide a simple adjustment to the standard errors which delivers asymptotically exact tests. In this sense, our inference results generalize those of \cite{bugni/canay/shaikh:2018} for individual-level randomized experiments to settings with cluster-level randomization.

By virtue of its sampling framework, our paper is distinct from a closely related and complementary literature that has analyzed cluster randomized experiments from a finite-population perspective. Important contributions to this literature include \cite{middleton2015unbiased}, \cite{athey2017econometrics}, \cite{hayes2017cluster}, \cite{de2020level}, \cite{schochet2021design}, \cite{su2021model}, and \cite{abadie2023should}.  The primary source of uncertainty in this literature is ``design-based'' uncertainty stemming from the randomness in treatment assignment, though parts of the literature additionally permit up to two additional sources of uncertainty: the randomness from sampling clusters from a finite population of clusters and the randomness from sampling only a subset of the finite number of units in each cluster.  In the context of such a sampling framework, the literature has defined finite-population counterparts to both our equally-weighted and size-weighted cluster-level average treatment effects.  See, for instance, \citet[Chapter 8]{athey2017econometrics} and \citet[Section 4]{su2021model}.  In particular, \cite{su2021model} provide estimators and methods for inference about each quantity when treatment is assigned completely at random (excluding, as a result, stratified covariate-adaptive treatment assignments). Our contribution may thus be viewed as leveraging our novel super-population sampling framework to develop results for general stratified covariate-adaptive randomization procedures: we discuss further comparisons between these approaches in Remark \ref{rem:finpop_contrast}. With this in mind, our results may be especially relevant for the analysis of data that are sampled from a well-defined larger population; \cite{muralidharan2017experimentation} argue that at least 30\% of the experiments they examined in economics satisfied this criterion. 



Our paper is also related to a large literature on the analysis of clustered data (not necessarily from experiments) in econometrics and statistics.  Prominent contributions to this literature include \cite{liang1986longitudinal}, \cite{hansen2007asymptotic}, \cite{hansen/lee:2019}, and \cite{djogbenou2019asymptotic}. Additional references can be found in the surveys \cite{cameron2015practitioner} and \cite{mackinnon2023cluster}.  These papers are designed as methods for inference for parameters defined via linear models or estimating equations rather than parameters like our equally-weighted or size-weighed cluster-level average treatment effects that are defined explicitly in terms of potential outcomes.  Importantly, in almost all of these papers, the sampling framework treats cluster sizes as non-random, though we note that in some cases the results are rich enough to permit the distribution of the data to vary across clusters: further discussion is provided in Remark \ref{rem:det_cluster}. Finally, none of these papers seem to explicitly consider the additional complications stemming from sampling only a subset of the units within each cluster. 

The remainder of our paper is organized as follows.  Section \ref{sec:setup} describes our setup and notation, including a formal description of our sampling framework and two parameters of interest.  We then propose in Section \ref{sec:main} estimators for each of these two quantities and develop the requisite distributional approximations to use them for inference about each quantity. In Section \ref{sec:sims}, we demonstrate the finite-sample behavior of our proposed estimators in a small simulation study. Finally, in Section \ref{sec:application}, we conduct an empirical exercise to demonstrate the practical relevance of our findings. Proofs of all results are included in the Appendix.

\section{Setup and Notation} \label{sec:setup}
\subsection{Notation and Sampling Framework}
Let $Y_{i,g}$ denote the (observed) outcome of the $i$th unit in the $g$th cluster, $A_g$ denote an indicator for whether the $g$th cluster is treated or not, $Z_g$ denote observed baseline covariates for the $g$th cluster, and $N_g$ the size of the $g$th cluster.  Further denote by $Y_{i,g}(1)$ the potential outcome of the $i$th unit in the $g$th cluster if treated and by $Y_{i,g}(0)$ the potential outcome of the $i$th unit in the $g$th cluster if not treated. As usual, the (observed) outcome and potential outcomes are related to treatment assignment by the relationship 
\begin{equation} \label{eq:obsY}
Y_{i,g} = Y_{i,g}(1)A_g + Y_{i,g}(0)(1 - A_g)~.
\end{equation}

We model the distribution of the data described above in two parts: a super-population sampling framework for the clusters and an assignment mechanism which assigns the clusters to treatments. The sampling framework itself can be described in two stages. In the first stage, an i.i.d.\ sample of $G$ clusters is drawn from a distribution of clusters. In the second stage, a subset of the individual units within each cluster is sampled.  A key feature of this framework is that the cluster size $N_g$ is modeled as a random variable in the same way as other cluster characteristics $Z_g$. While the clusters are (ex-ante) identically distributed, we note that they may exhibit heterogeneity in terms of their (ex-post) realizations of $N_g$ and $Z_g$. The second sampling stage allows for settings in which the analyst does not observe all of the units within a cluster. Define $\mathcal{M}_g$ to be the subset of $\{1, \ldots, N_g\}$ corresponding to the observations within the $g$th cluster that are sampled by the researcher. We emphasize that a realization of $\mathcal{M}_g$ is a \emph{set} whose cardinality we denote by $|\mathcal{M}_g|$, whereas a realization of $N_g$ is a positive integer. For example, in the event that all observations in a cluster are sampled, $\mathcal{M}_g = \{1, \ldots, N_g\}$ and $|\mathcal{M}_g| = N_g$. Once the sample of clusters is realized, the experiment assigns treatments $A^{(G)} := (A_g : 1 \leq g \leq G)$ using an assignment rule that stratifies according to baseline covariates $Z_g$ and cluster sizes $N_g$. Formally, denote by $P_G$ the distribution of the observed data $$(((Y_{i,g} : i \in \mathcal{M}_g), A_g, Z_g, N_g) : 1 \leq g \leq G )$$ that arises from sampling and treatment assignment, and by $ Q_G$ the distribution of $$W^{(G)} := (((Y_{i,g}(1),Y_{i,g}(0) : 1 \le i \le N_g),\mathcal{M}_g, Z_g, N_g) : 1 \leq g \leq G )~.$$  Note that the observed distribution $P_G$ is determined jointly by \eqref{eq:obsY} together with the distribution of $A^{(G)}$ and $Q_G$, so we will state our assumptions below in terms of these two quantities. 

We begin by describing our assumptions on the distribution of $A^{(G)}$. Strata are constructed from the observed, baseline covariates $Z_g$ and cluster sizes $N_g$ using a function $S: \text{supp}((Z_g, N_g)) \rightarrow \mathcal{S}$, where $\mathcal{S}$ is a finite set. For $1 \le g \le G$, let $S_g = S(Z_g, N_g)$ and denote by $S^{(G)}$ the vector of strata $(S_1, S_2, \ldots, S_G)$. In what follows, we rule out trivial strata by assuming that $p(s) := P\{S_g = s\} > 0$ for all $s \in \mathcal{S}$.  For $s \in \mathcal{S}$, let
\begin{equation}\label{eqn:D_G(s)}
    D_G(s) := \sum_{1\leq g \leq G} (I\{A_g = 1\} - \pi) I\{S_g = s\},
\end{equation}
where $\pi \in (0,1)$ is the “target” proportion of clusters to assign to treatment in each stratum. Note that $D_G(s)$ measures the amount of imbalance in stratum $s$ relative to the target proportion $\pi$. Our requirements on the treatment assignment mechanism are then summarized as follows:

\begin{assumption}\label{ass:assignment}
The treatment assignment mechanism is such that
\begin{enumerate}
    \item[(a)] $W^{(G)} \independent A^{(G)} | S^{(G)}$
    \item[(b)] $\left\{ \left\{\frac{D_G(s)}{\sqrt{G}} \right\}_{s \in \mathcal{S}}  \Big| S^{(G)}\right\} \xrightarrow{d} N(0, \Sigma_D)$ a.s., where
    \begin{equation*}
        \Sigma_D = \text{diag}\{p(s)\tau(s): s\in\mathcal{S} \}
    \end{equation*}
    with $0\leq \tau(s) \leq \pi (1-\pi)$ for all $s\in\mathcal{S}$.
\end{enumerate}
\end{assumption}

Assumption \ref{ass:assignment} mirrors the assumption on assignment mechanisms considered in \cite{bugni/canay/shaikh:2018} for individual-level randomized experiments. Assumption \ref{ass:assignment}.(a) requires that the assignment mechanism is a function only of the strata and an exogenous randomization device. Assumption \ref{ass:assignment}.(b) requires that the randomization mechanism assigns treatments within each stratum so that the fraction of units being treated has a well-behaved limiting distribution centered around the target proportion $\pi$. For each stratum $s \in \mathcal{S}$, the parameter $\tau(s) \in [0,1]$ determines the amount of dispersion that the treatment assignment mechanism allows on the fraction of units assigned to the treatment in that stratum. A lower value of $\tau(s)$ implies that the treatment assignment mechanism imposes a higher degree of ``balance'' or ``control'' of the treatment assignment proportion relative to its desired target value. \cite{bugni/canay/shaikh:2018} provide several important examples of assignment mechanisms satisfying this assumption which are used routinely in economics. In particular, Assumption \ref{ass:assignment} is satisfied by stratified block randomization \citep[see, e.g.,][]{angelucci2015microcredit, attanasio2015impacts, duflo2015education}, which assigns exactly a fraction $\pi$ of units within each stratum to treatment, at random. When $\tau(s) = 0$ (as is the case for stratified block randomization) for all $s \in \mathcal{S}$, we say that the assignment mechanism achieves ``strong balance." Note further that the assumption also applies in settings without stratification, in which case $|\mathcal{S}| = 1$. Despite its broad applicability, Assumption \ref{ass:assignment} nevertheless precludes treatment assignment mechanisms with many ``small" strata, by the virtue of assuming that $\mathcal{S}$ is a fixed finite set. This precludes, for instance, ``matched pairs'' designs \citep[see, e.g.,][]{banerjee2015miracle,crepon2015estimating}, the analysis of which can be found in the companion paper \cite{bai2022pairs}.

We now describe our assumptions on $Q_G$.  In order to do so, it is useful to introduce some further notation.  To this end, define $R_G(\mathcal{M}^{(G)}, Z^{(G)}, N^{(G)})$ to be the distribution of $$((Y_{i,g}(1), Y_{i,g}(0) : 1 \le i \le N_g) : 1 \leq g \leq G) ~\big |~ \mathcal{M}^{(G)}, Z^{(G)}, N^{(G)}~,$$ where $\mathcal{M}^{(G)} := (\mathcal{M}_g : 1 \leq g \leq G)$, $Z^{(G)} := (Z_g : 1 \leq g \leq G)$ and $N^{(G)} := (N_g : 1 \leq g \leq G)$.  Note that $Q_G$ is completely determined by $R_G(\mathcal{M}^{(G)}, Z^{(G)}, N^{(G)})$ and the distribution of $(\mathcal{M}^{(G)}, Z^{(G)}, N^{(G)})$.  
While $N_g$ obviously determines the length of $(Y_{i,g}(1), Y_{i,g}(0) : 1 \le i \le N_g)$, we emphasize that it may affect the location and shape of its distribution as well. In this paper, we say cluster sizes are {\it ignorable} whenever the individual-level average treatment effect does not depend on the cluster size, in the sense that
\begin{equation}
    P\{E[Y_{i,g}(1) - Y_{i,g}(0)|N_g] = E[Y_{i,g}(1) - Y_{i,g}(0)]\text{ for all }1\leq i \leq N_g\}=1~~\text{ for all }1\leq g \leq G.
    \label{eq:ignorableDefn}
\end{equation}
Consequently, we say that cluster sizes are \emph{non-ignorable} whenever \eqref{eq:ignorableDefn} fails. Example \ref{ex:numerical} provides a simple illustration of non-ignorability, and Remark \ref{rem:comparingestimands} provides some related discussion of the consequences of assuming clusters are, in fact, ignorable.
Finally, for $a \in \{0,1\}$, define $$\bar Y_g(a) := \frac{1}{|\mathcal{M}_g|} \sum_{i \in \mathcal{M}_g} Y_{i,g}(a)~.$$  The following assumption states our requirements on $Q_G$ using this notation.

\begin{assumption} \label{ass:QG}
The distribution $Q_G$ is such that \vspace{-.25cm}
\begin{enumerate}[(a)]
\item $\{(\mathcal{M}_g,Z_g,N_g), 1 \leq g \leq G\}$ is an i.i.d.\ sequence of random variables.
\item The distribution of $((Y_{i,g}(1), Y_{i,g}(0) : 1 \le i \le N_g) : 1 \leq g \leq G) ~\big |~ \mathcal{M}^{(G)}, Z^{(G)}, N^{(G)}$ can be factored as $$R_G(\mathcal{M}^{(G)}, Z^{(G)}, N^{(G)}) = \prod_{1 \leq g \leq G} R(\mathcal{M}_g,Z_g,N_g)~,$$
where $R\left( \mathcal{M}_{g},Z_{g},N_{g}\right) $ denotes the distribution of $(Y_{i,g}(1),Y_{i,g}(0): 1 \le i \le N_g)$ conditional on $\{\mathcal{M} _{g},Z_{g},N_{g} \}$.
\item $\mathcal{M}_g \independent (Y_{i,g}(1),Y_{i,g}(0) : 1 \leq i \leq N_g) ~ \big | ~ Z_g, N_g$ for all $1 \leq g \leq G$, that is, $R(\mathcal{M}_g,Z_g,N_g)=R(Z_g,N_g)$.
\item For $a \in \{0,1\}$ and $1 \leq g \leq G$, $$E[\bar Y_g(a)|N_g] = E\left [\frac{1}{N_g}\sum_{1 \leq i \leq N_g} Y_{i,g}(a) \Big |N_g \right] ~\text{w.p.1} ~.$$

\item $P\{|\mathcal{M}_g| \geq 1\} = 1$ and $E[N_g^2] < \infty$.
\item For some $C < \infty$, $P\{E[Y^2_{i,g}(a)| N_g, Z_g] \leq C \text{ for all } 1 \leq i \leq N_g \} = 1$ for all $a \in \{0,1\}$ and $1 \leq g \leq G$.
\end{enumerate}
\end{assumption}

\noindent Assumptions \ref{ass:QG}.(a)--(b) formalize the idea that our data consist of an i.i.d.\ sample of clusters, where the cluster sizes are themselves random and possibly related to potential outcomes. An important implication of these two assumptions for our purposes is that 
\begin{equation} \label{eq:iidclusters}
\left \{(\bar{Y}_g(1), \bar{Y}_g(0), |\mathcal{M}_g|, Z_g, N_g \right), 1 \leq g \leq G\}
\end{equation}
is an i.i.d.\ sequence of random variables, as established by Lemma \ref{lemma:iid} in the Appendix. 

Assumptions \ref{ass:QG}.(c)--(d) impose high-level restrictions on the second stage of the sampling framework. Assumption \ref{ass:QG}.(c) allows the subset of observations sampled by the experimenter to depend on $Z_g$ and $N_g$, but rules out dependence on the potential outcomes within the cluster itself.
Assumption \ref{ass:QG}.(d) is a high-level assumption which guarantees that we can extrapolate from the observations that are sampled to the observations that are not sampled. Note that Assumptions \ref{ass:QG}.(c)--(d) are trivially satisfied whenever $\mathcal{M}_g = \{1, \ldots, N_g\}$ for all $1 \leq g \leq G$ with probability one, i.e., whenever all observations within each cluster are always sampled. Assumption \ref{ass:QG}.(d) is also satisfied whenever Assumption \ref{ass:QG}.(c) holds and there is sufficient homogeneity across the observations within each cluster in the sense that $P\{ E[Y_{i,g}(a)|N_g, Z_g] = E[Y_{j,g}(a)|N_g, Z_g] \text{ for all } 1 \leq i, j \leq N_g \} = 1$ for $a \in \{0,1\}$. Finally, we show in Lemma \ref{prop:sample} below that if $\mathcal{M}_g$ is drawn as a random sample without replacement from $\{1, 2, \ldots, N_g\}$ in an appropriate sense, then Assumptions \ref{ass:QG}.(c)--(d) are also satisfied.

Assumptions \ref{ass:QG}.(e)--(f) impose some mild regularity on the (conditional) moments of the distribution of cluster sizes and potential outcomes, in order to permit the application of relevant laws of large numbers and central limit theorems. Note that Assumption \ref{ass:QG}.(e) does not rule out the possibility of observing arbitrarily large clusters but does place restrictions on the frequency of extremely large realizations. For instance, two consequences of Assumptions \ref{ass:QG}.(a) and (e) are that\footnote{The first is an immediate consequence of the law of large numbers and the Continuous Mapping Theorem. The second follows from Lemma S.1.1 in \cite{bai2021inference}.}
\[\frac{\sum_{1\le g \le G}N_g^2}{\sum_{1\le g\le G}N_g} = O_P(1)~,\]
and
\[\frac{\max_{1\le g \le G}N_g^2}{\sum_{1 \le g \le G}N_g} \xrightarrow{P} 0~,\]
which mirror heterogeneity restrictions imposed in earlier work on the analysis of clustered data when cluster sizes are modeled as non-random \citep[see, for example, Assumption 2 in][]{hansen/lee:2019}, but are modestly stronger than the heterogeneity restrictions considered in recent work studying cluster randomized experiments from a finite-population perspective \citep[see, for instance, Theorem 1 in][]{su2021model}. We use Assumption \ref{ass:QG}(e) extensively when establishing asymptotic normality in Theorem \ref{thm:mainSECT}; recent work by \cite{sasaki2022non} and \cite{chiang2023cluster}, however, suggests that one may be able to sometimes obtain asymptotic normality even when $E[N_g^2] = \infty$, provided that certain delicate conditions about the tail behavior of $N_g$ are satisfied. 

\begin{lemma}\label{prop:sample}
Suppose that $\mathcal{M}_g \independent (Y_{i,g}(1),Y_{i,g}(0) : 1 \leq i \leq N_g) ~ \big | ~ Z_g, N_g$ for all $1 \le g \le G$, and that, conditionally on $(Z_g, N_g, |\mathcal{M}_g|)$, $\mathcal{M}_g$ is drawn uniformly at random from all possible subsets of size $|\mathcal{M}_g|$ from $\{1, 2, \ldots. N_g\}$. Then, Assumptions \ref{ass:QG}.(c)--(d) are satisfied.
\end{lemma}

\begin{remark}\label{rem:drifting}
We could in principle modify our framework so that the distribution of cluster sizes is allowed to depend on the number of clusters $G$.  By doing so, we would be able to weaken Assumption \ref{ass:QG}.(e) at the cost of strengthening Assumption \ref{ass:QG}.(f) to require, for example, uniformly bounded $2+\delta$ moments for some $\delta > 0$. Such a modification, however, would complicate the exposition and the resulting procedures would ultimately be the same.  We therefore see no apparent benefit and do not pursue it further in this paper.
\end{remark}

\begin{remark}\label{rem:det_cluster}
   An attractive feature of our framework is that, by virtue of modeling cluster sizes as random, it is straightforward to permit dependence between the cluster size and other features of the cluster, such as the distribution of potential outcomes within the cluster.  In this way, our setting departs from other frameworks in the literature on clustered data in which the cluster sizes are treated as deterministic: see, for example, \cite{hansen/lee:2019}.  We note, however, that the results in this literature permit the distribution of the data across clusters to be non-identically distributed, and in fact the literature has noted that the method described in \cite{liang1986longitudinal} may fail when cluster sizes are non-ignorable: see, in particular, \cite{benhin2005mean}. Although it is possible that these related papers could be applied to our framework by first conditioning on the cluster sizes, results obtained in this way would necessarily hold conditionally on the cluster sizes, whereas our results are unconditional.  
\end{remark}

\subsection{Parameters of Interest}\label{sec:parameters}
In settings with cluster data, there are multiple ways to aggregate, or average, the heterogeneous treatment effect $Y_{i,g}(1) - Y_{i,g}(0)$. In particular, an important consideration is whether the units of interest are the clusters themselves or the individuals within the cluster. This distinction is precisely what motivates the two parameters of interest we study in this paper. Before introducing these parameters, we present a simple example that will help us illustrate the differences. 

\begin{myexa}\label{ex:numerical}
We revisit the example described in the introduction, where clusters represent primary care facilities, or clinics, within a state, as in the empirical application in Section \ref{sec:application}. Let's suppose there are two types of clinics potentially exposed to a treatment: ``big'' clinics with $N_{g}=40$ regular patients, and ``small'' clinics with $N_{g}=10$ regular patients. Suppose further that cluster size is non-ignorable in that $Y_{i,g}\left(1\right) -Y_{i,g}\left( 0\right) =1$ for all patients in a ``big'' clinic, $Y_{i,g}\left( 1\right) -Y_{i,g}\left( 0\right) =-2$ for all patients
in a ``small'' clinic, and that both types of clinics are equally likely, i.e., 
\begin{equation*}
P\left\{N_{g}=40\right\} =P\left\{ N_{g}=10\right\} =1/2~.
\end{equation*}
Policy makers evaluating the adoption of the treatment may arrive at different conclusions depending on their objective. For instance, if clinics are considered as the unit of analysis, a policy maker may perceive the treatment as harmful, as half of the clinics experience a positive treatment effect of $1$, and the other half experience a negative treatment effect of $-2$, resulting in an ``average'' treatment effect of $-1/2$. Conversely, if individuals (patients) within the state are considered as the unit of analysis, the treatment may appear beneficial. In this case, $4/5$ of the patients experience a positive treatment effect of $1$, while only $1/5$ experience a negative treatment effect of $-2$, resulting in an ``average'' treatment effect of $2/5$. The distinction between these two different ways of averaging motivates the formal definitions of the parameters that we consider below.
\end{myexa}

Example \ref{ex:numerical} illustrates that characterizing an ``average'' treatment effect in settings with cluster data depends on whether clusters or individuals are the focus of the analysis. This distinction, as illustrated in the example, is relevant whenever clusters feature cluster size heterogeneity and average treatment effect heterogeneity with respect to cluster size (even if the treatment effect is assumed to be homogeneous within clusters as in Example \ref{ex:numerical}). When there is average treatment effect heterogeneity within and across clusters, the question becomes whether the aggregated information in the clusters, i.e., $(1/N_g)\sum_{1 \leq i \leq N_g} Y_{i,g}(a)$, should be weighted by the size of the cluster or not. 

Motivated by Example \ref{ex:numerical}, we consider two different parameters of interest: one that considers the clusters as the units of interest, and one that considers the individuals as the units of interest. Both of these parameters can be written in the form 
\begin{equation} \label{eq:general-weights}
E\left [\omega_g \left ( \frac{1}{N_g} \sum_{1 \leq i \leq N_g} \left(Y_{i,g}(1) - Y_{i,g}(0)\right) \right ) \right ]
\end{equation}
for different choices of (possibly random) weights $\omega_g, 1 \leq g \leq G$ satisfying $E[\omega_g] = 1$.  The first parameter of interest corresponds to the choice of $\omega_g = 1$, thus weighting the average effect of the treatment across clusters equally:
\begin{equation} \label{eq:ECA}
\theta_1(Q_G) := E\left [ \frac{1}{N_g} \sum_{1 \leq i \leq N_g} \left(Y_{i,g}(1) - Y_{i,g}(0)\right) \right ]~.
\end{equation}
We refer to this quantity as the equally-weighted cluster-level average treatment effect. $\theta_1(Q_G)$ can be thought of as the average treatment effect where the clusters themselves are the units of interest. The second parameter of interest corresponds to the choice of $\omega_g = N_g/E[N_g]$, thus weighting the average effect of the treatment across clusters in proportion to their size:
\begin{equation} \label{eq:SECT}
\theta_2(Q_G) := E\left [ \frac{1}{E[N_g]} \sum_{1 \leq i \leq N_g} \left(Y_{i,g}(1) - Y_{i,g}(0)\right) \right ]~.
\end{equation}
We refer to this quantity as the size-weighted cluster-level average treatment effect. $\theta_2(Q_G)$ can be thought of as the average treatment effect where individuals are the units of interest. Note that Assumptions \ref{ass:QG}.(a)--(b) imply that we may express both $\theta_1(Q_G)$ and $\theta_2(Q_G)$ as a function of $R$ and the common distribution of $(\mathcal{M}_g,Z_g,N_g)$.  In particular, neither quantity depends on $g$ or $G$ and so in what follows we simply denote $\theta_1 = \theta_1(Q_G)$, $\theta_2 = \theta_2(Q_G)$. 

If treatment effects are heterogeneous with respect to cluster size and cluster sizes are themselves heterogeneous, then $\theta_1$ and $\theta_2$ are indeed distinct parameters. We illustrate this in the context of Example \ref{ex:numerical}. 

\begin{myexacont} 
Recall the setting of Example \ref{ex:numerical}. The equally-weighted cluster-level average treatment effect simply equals 
\begin{align*}
\theta_1 &= P\{N_g = 10\}E\left[ \frac{1}{N_{g}}\sum_{1 \le i \le N_g}(Y_{i,g}\left( 1\right)
-Y_{i,g}\left( 0\right)) ~\Biggr |~N_{g}=10\right] \\
&\qquad + P\{N_g = 40\}E\left[ \frac{1}{N_{g}}\sum_{1 \le i \le N_g}(Y_{i,g}\left( 1\right)
-Y_{i,g}\left( 0\right)) ~\Biggr |~N_{g}=40\right]\\
&= \left( 1/2\right) \times -2+\left( 1/2\right) \times 1=-1/2~.
\end{align*}%
The parameter $\theta_1$ captures an average treatment effect where the clusters are the units of interest since both treatment effects $1$ and $-2$ receive the same weight (both types of clinics are equally likely). The size-weighted cluster-level average treatment effect, in turn, equals
\begin{align*}
\theta_2 &= \frac{P\{N_g = 10\}}{E[N_g]}E\left[\sum_{1 \le i \le N_g}(Y_{i,g}\left( 1\right)
-Y_{i,g}\left( 0\right)) ~\Biggr |~N_{g}=10\right] \\
&\qquad + \frac{P\{N_g = 40\}}{E[N_g]}E\left[\sum_{1 \le i \le N_g}(Y_{i,g}\left( 1\right)
-Y_{i,g}\left( 0\right)) ~\Biggr |~N_{g}=40\right]\\
&= \frac{(1/2)}{25}\times -20 + \frac{(1/2)}{25}\times 40 =2/5~.
\end{align*}%
The parameter $\theta_2$ captures an average treatment effect where the individuals are the units of interest since both treatment effects, $1$ and $-2$, are weighted by the proportion of the patients in the state that attend each type of clinic. 
\end{myexacont}

\begin{remark} \label{rem:comparingestimands}
While we generally expect $\theta_1$ and $\theta_2$ to be distinct, they are equivalent in some special cases. For example, if all clusters are of the same fixed size $k$, i.e., $P\{N_g = k\} = 1$, then it follows immediately that $\theta_1 = \theta_2$. Alternatively, if treatment effects are constant, so that $P\{Y_{i,g}(1) - Y_{i,g}(0) = \tau \text{ for all } 1 \leq i \leq N_g\} = 1$, then $\theta_1 = \theta_2$. Beyond these two cases, we have $\theta _{1}=\theta _{2}$ whenever cluster sizes are ignorable in the sense of \eqref{eq:ignorableDefn} and the average treatment effects are homogeneous in the sense that $P\{E[Y_{i,g}(1) - Y_{i,g}(0)] = E[Y_{j,g}(1) - Y_{j,g}(0)]$ for all $1 \le i,j \le N_g\} = 1$.  Note that this last statement is not generally true if one replaces \eqref{eq:ignorableDefn} with 
\[P\{E[Y_{i,g}(1) - Y_{i,g}(0)|N_g, X_g] = E[Y_{i,g}(1) - Y_{i,g}(0)|X_g]\text{ for all }1\leq i \leq N_g\}=1 \text{ for all } 1\leq g \leq G~,\]
where $X_g$ represents a collection of cluster-level characteristics (either observed or unobserved). In this sense, even if $N_g$ is simply a proxy for other characteristics $X_g$, it still plays an important role in our analysis through the distinction between $\theta_1$ and $\theta_2$.
\end{remark}

Not all estimands that arise from commonly used empirical strategies take the form in \eqref{eq:general-weights}. For example, as we show in Theorem \ref{theorem:limOLS} in the next section, the usual difference-in-means estimator consistently estimates the following population parameter, 
\begin{equation}\label{eq:limit-DIM}
        \vartheta := E\left [ \frac{1}{E[|\mathcal{M}_g|]} \sum_{i \in \mathcal{M}_g} \left(Y_{i,g}(1) - Y_{i,g}(0)\right) \right ]~.  
\end{equation}
This parameter corresponds to a \emph{sample}-weighted cluster-level average treatment effect and, without assumptions on the sampling process, does not generally identify either an average treatment effect when the clusters are the units of interest or an average treatment effect when the individuals are the units of interest. In other words, when treatment effects are heterogeneous and cluster sizes are non-ignorable,  $\vartheta$ need not equal either $\theta_1$ defined in \eqref{eq:ECA} or $\theta_2$ defined in \eqref{eq:SECT}. Two specific sampling designs for which $\vartheta$ equals either $\theta_1$ or $\theta_2$ are worth highlighting. First, if $P\{|\mathcal{M}_g| = k\} = 1$ for all $1 \leq g \leq G$, then $\vartheta$ is equal to $\theta_1$. This is intuitive, as in this case, $\vartheta$ gives equal weights to each cluster and thus behaves as if the units of interest are the clusters themselves. Second, if $|\mathcal{M}_g|$ is a constant fraction of $N_g$ for all $1 \leq g \leq G$,  i.e., $P\{|\mathcal{M}_g| = \gamma N_g \} = 1$ for some $0 < \gamma \le 1$,  then $\vartheta$ is equal to $\theta_2$. This is also intuitive, as in this case, the relative weights of individuals in the parameter $\vartheta$ coincide with the weights they would have obtained if the units of interest were the individuals. In general, as illustrated in the following example, $\vartheta$ may not equal either $\theta_1$ or $\theta_2$.

\begin{myexacont}
Recall the setting of Example \ref{ex:numerical}. Suppose further that the experimenter samples $|\mathcal{M}_g| = 5$ patients at random without replacement from each ``small'' clinic, and $|\mathcal{M}_g| = 10$ patients at random without replacement from each ``big'' clinic. Importantly, note that this sampling scheme does not maintain the relative proportions of these clinics that exist at the population level. It is now straightforward to show that 
\begin{align*}
\vartheta &= P\{N_g = 10\}\frac{E[|\mathcal{M}_g| ~|~N_g=10]}{E[|\mathcal{M}_g|]}E\left[\frac{1}{N_g}\sum_{1 \le i \le N_g}(Y_{i,g}\left( 1\right)-Y_{i,g}\left( 0\right)) ~\Biggr |~N_{g}=10\right] \\
&\qquad + P\{N_g = 40\}\frac{E[|\mathcal{M}_g| ~|~N_g=40]}{E[|\mathcal{M}_g|]}E\left[\frac{1}{N_g}\sum_{1 \le i \le N_g}(Y_{i,g}\left( 1\right)
-Y_{i,g}\left( 0\right)) ~\Biggr |~N_{g}=40\right]\\
&= (1/2)\times \frac{5}{15/2}\times -2 + (1/2)\times\frac{10}{15/2}\times 1 =0~,
\end{align*}%
where the first equality exploits Assumption \ref{ass:QG}(c), so that $|\mathcal{M}_g|$ and $(Y_{i,g}\left( 1\right)-Y_{i,g}\left( 0\right))$ are independent conditional on $N_g$, and Assumption \ref{ass:QG}(d), so that $\frac{1}{|\mathcal{M}_g|}\sum_{i \in \mathcal M_g}$ can be replaced by $\frac{1}{N_g}\sum_{1 \le i \le N_g}$. We conclude that $\vartheta$ is not equal to either $\theta_1$ or $\theta_2$. Indeed, it is straightforward to show that as we vary the distribution of $|\mathcal{M}_g|$ conditional on $N_g$, $\vartheta$ could take any value between $-1.72$ and $0.93$ in this example, so that it could be smaller, in-between, or larger than both $\theta_1=-1/2$ and $\theta_2=2/5$.
\end{myexacont}

\section{Main Results} \label{sec:main}
\subsection{Asymptotic Behavior of the Difference-in-Means Estimator}\label{sec:limOLS}

Given its central role in the analysis of randomized experiments, we begin this section by studying the asymptotic behavior of the difference-in-means estimator 
\begin{equation} \label{eq:theta1hat}
\hat \theta^{\rm alt}_{G} := \frac{\sum_{1 \leq g \leq G} \sum_{i \in \mathcal{M}_g} Y_{i,g} A_g}{\sum_{1 \leq g \leq G} |\mathcal{M}_g| A_g} - \frac{\sum_{1 \leq g \leq G} \sum_{i \in \mathcal{M}_g} Y_{i,g} (1- A_g)}{\sum_{1 \leq g \leq G} |\mathcal{M}_g| (1- A_g)}~.{}
\end{equation}
Note that $\hat \theta^{\rm alt}_{G}$ may be obtained as the estimator of the coefficient on $A_g$ in the following ordinary least squares regression:
$$ \texttt{regress } Y_{i,g} \texttt{ on } \text{constant} +  A_g~.$$

The following theorem derives the probability limit of this estimator:
\begin{theorem} \label{theorem:limOLS}
Under Assumptions \ref{ass:assignment} and \ref{ass:QG},
\[\hat \theta^{\rm alt}_{G} \stackrel{p}{\rightarrow} E\left [ \frac{1}{E[|\mathcal{M}_g|]} \sum_{i \in \mathcal{M}_g} \left(Y_{i,g}(1) - Y_{i,g}(0)\right) \right ] = \vartheta\] 
as $G \rightarrow \infty$.
\end{theorem}
As we discussed in the previous section, the quantity $\vartheta$ corresponds to a \emph{sample}-weighted cluster-level average treatment effect. In general, this parameter will not equal either $\theta_1$ or $\theta_2$ as illustrated in Example \ref{ex:numerical}.  As a result, unless the experimenter is interested in a distinct weighting of the cluster-level treatment effects that differs from those that arise when the clusters themselves are the units of interest or when the individuals within the clusters are the units of interest, care must be taken when interpreting $\hat \theta^{\rm alt}_{G}$. While it is true, as we discussed in Section \ref{sec:parameters}, that $\vartheta$ is in fact equal to either $\theta_1$ or $ \theta_2$ for some specific sampling designs, in what follows we consider alternative estimators which are generally consistent for $\theta_1$ and $\theta_2$ without imposing additional restrictions on the sampling procedure.

\subsection{Equally-weighted Cluster-level Average Treatment Effect}\label{sec:equal}

In this section, we consider the estimation of $\theta_1$ defined in \eqref{eq:ECA}.  To this end, consider the following difference-in-``average of averages'' estimator:
\begin{equation} \label{eq:theta2hat}
\hat \theta_{1,G} := \frac{\sum_{1 \leq g \leq G} \bar Y_g A_g}{\sum_{1 \leq g \leq G} A_g} - \frac{\sum_{1 \leq g \leq G} \bar Y_g (1- A_g)}{\sum_{1 \leq g \leq G} (1- A_g)}~,
\end{equation}
where 
\begin{equation} \label{eq:barYg}
\bar Y_g = \frac{1}{|\mathcal{M}_g|}\sum_{i \in \mathcal{M}_g} Y_{i,g}~.
\end{equation}
For what follows, it will be useful to introduce some notation to denote various types of averages. Given a sequence of random variables $(C_g: 1 \le g \le G)$, consider the following definitions:
\begin{eqnarray*}
\hat{\mu}^C_{G,a}(s) &:=& \frac{1}{\sum_{1 \le g \le G}I\{A_g = a, S_g = s\}}\sum_{1 \le g \le G}C_gI\{A_g = a, S_g = s\}~,\\
\hat{\mu}^C_G(s) &:=& \frac{1}{\sum_{1 \le g \le G}I\{S_g = s\}}\sum_{1 \le g \le G}C_gI\{S_g = s\}~,\\
\hat{\mu}^C_{G,a} &:=& \frac{1}{\sum_{1 \le g \le G}I\{A_g = a\}}\sum_{1 \le g \le G}C_gI\{A_g = a\}~.
\end{eqnarray*}
Given this notation, $\hat{\theta}_{1,G}$ could alternatively be written as
\[\hat{\theta}_{1,G} = \hat{\mu}^{\bar{Y}}_{G,1} - \hat{\mu}^{\bar{Y}}_{G,0}~.\]
Note that $\hat \theta_{1,G}$ may be obtained as the estimator of the coefficient on $A_g$ in the following ordinary least squares regression:
$$ \texttt{regress } \bar Y_{g} \texttt{ on } \text{constant} +  A_g~.$$
As such, we can view $\hat{\theta}_{1,G}$ as estimating the treatment effect for an individual-level randomized experiment where the clusters are themselves the units of interest. The following theorem derives the asymptotic behavior of this estimator.

\begin{theorem} \label{theorem:mainECA}
Under Assumptions \ref{ass:assignment} and \ref{ass:QG}, $$\sqrt G (\hat \theta_{1,G} - \theta_1) \stackrel{d}{\rightarrow} N(0,\sigma_1^2)$$ as $G \rightarrow \infty$, where 
\begin{eqnarray*}
\sigma_1^2 := &\frac{1}{\pi} {\text Var}[\bar Y^{\dagger}_g(1)] + \frac{1}{1-\pi} {\text Var}[\bar Y^{\dagger}_g(0)] + E[(\bar{m}_1(S_g) - \bar{m}_0(S_g))^2] + E\left[\tau(S_g)\left(\frac{1}{\pi}\bar{m}_1(S_g) + \frac{1}{1-\pi}\bar{m}_0(S_g)\right)^2\right]~,
\end{eqnarray*}
with 
\begin{eqnarray}
\bar{Y}^{\dagger}_g(a) &:=& \bar{Y}_g(a) - E[\bar{Y}_g(a)|S_g]~, \nonumber \\
\bar{m}_a(S_g) &:=& E[\bar{Y}_g(a)|S_g] - E[\bar{Y}_g(a)]~, \label{eq:barm}
\end{eqnarray}
and $\pi$, $\tau(\cdot)$ defined as in Assumption \ref{ass:assignment}.
\end{theorem}

From this result it is immediate that $\hat{\theta}_{1,G}$ is most efficient when paired with an assignment mechanism that features $\tau(s) = 0$  for every $s \in \mathcal{S}$ (i.e., strong balance) and least efficient when $\tau(s) = \pi(1 - \pi)$ for every $s \in \mathcal{S}$. The next result shows that, as a consequence, the probability limit of the standard heteroskedasticity-robust variance estimator is generally too large relative to $\sigma^2_1$. 

\begin{theorem}\label{thm:limHC}
Let $\tilde \sigma_{1,G}^2$ denote the heteroskedasticity-robust estimator of the variance of the coefficient of $A_g$ in an ordinary least squares regression of $\bar{Y}_g$ on a constant and $A_g$. Note that this estimator can be written as 
$$\tilde \sigma_{1,G}^2 := \frac{1}{\frac{1}{G}\sum_{1 \leq g \leq G} A_g} \widehat {\text Var}[\bar Y_g(1)] + \frac{1}{\frac{1}{G}\sum_{1 \leq g \leq G} 1-A_g} \widehat {\text Var}[\bar Y_g(0)]~,$$ 
where $$\widehat {\text Var}[\bar Y_g(a)] := \hat{\mu}^{\bar{Y}^2}_a - (\hat{\mu}^{\bar{Y}}_a)^2~.$$
Then under Assumptions \ref{ass:assignment} and \ref{ass:QG}, 
\begin{equation}
    \tilde \sigma_{1,G}^2 \xrightarrow{p} \frac{1}{\pi}{\text Var}[\bar{Y}_g(1)] + \frac{1}{1 - \pi}{\text Var}[\bar{Y}_g(0)] \ge \sigma^2_1~,
    \label{eq:limHC}
\end{equation}
with equality if and only if 
\[(\pi(1 - \pi) - \tau(s))\left(\frac{1}{\pi}\bar{m}_1(s) + \frac{1}{1 - \pi}\bar{m}_0(s)\right)^2 = 0~,\] for every $s \in \mathcal{S}$.
\end{theorem}

Note that it can be shown that in the case of Bernoulli random assignment, where $A^{(G)}$ is an i.i.d.\ sequence with $P(A_g = 1) = \pi$, Assumption \ref{ass:assignment} is satisfied with $\tau(s) = \pi(1-\pi)$ for every $s \in \mathcal{S}$. As such, we obtain from Theorem \ref{thm:limHC} that in this case $\tilde \sigma_{1,G}^2$ is a consistent estimator of $\sigma^2_1$. Note that $\tilde \sigma_{1,G}^2$ is also consistent when there is no stratification (i.e., $|\mathcal{S}|=1$), since in this case $\bar{m}_a(s) = 0$ for $a \in \{0, 1\}$.

To facilitate the use of Theorem \ref{theorem:mainECA} for inference about $\theta_1$, we now provide an estimator of $\sigma_1^2$ which is consistent.  
For any $s \in \mathcal{S}$, let
\begin{equation}
    G(s) :=~\sum_{g\in G} 1\{S_g = s\}~.
    \label{eq:Gs}
\end{equation}
Then, define the following estimators:
\begin{eqnarray}
\hat{\zeta}^2_{\bar{Y}}(\pi) &:=& \frac{1}{\pi}\left(\hat{\mu}^{\bar{Y}^2}_{G,1} - \sum_{s \in \mathcal{S}}\frac{G(s)}{G}\hat{\mu}^{\bar{Y}}_{G,1}(s)^2\right) + \frac{1}{1 - \pi}\left(\hat{\mu}^{\bar{Y}^2}_{G,0} - \sum_{s \in \mathcal{S}}\frac{G(s)}{G}\hat{\mu}^{\bar{Y}}_{G,0}(s)^2\right) \label{eq:zetabarY} \\
\hat{\zeta}^2_H &:=& \sum_{s \in \mathcal{S}}\frac{G(s)}{G}\left(\left(\hat{\mu}^{\bar{Y}}_{G,1}(s) - \hat{\mu}^{\bar{Y}}_{G,1}\right) - \left(\hat{\mu}^{\bar{Y}}_{G,0}(s) - \hat{\mu}^{\bar{Y}}_{G,0}\right)\right)^2 \label{eq:zetaH}  \\
\hat{\zeta}^2_{A}(\pi) &:=& \sum_{s \in \mathcal{S}}\tau(s)\frac{G(s)}{G}\left(\frac{1}{\pi}\left(\hat{\mu}^{\bar{Y}}_{G,1}(s) - 
\hat{\mu}^{\bar{Y}}_1\right) + \frac{1}{1-\pi}\left(\hat{\mu}^{\bar{Y}}_{G,0}(s) - \hat{\mu}^{\bar{Y}}_{G,0}\right)\right)^2~, \label{eq:zetaA} 
\end{eqnarray}
and define $\hat{\sigma}^2_{1,G} :=  \hat{\zeta}^2_{\bar{Y}}(\pi) + \hat{\zeta}^2_H + \hat{\zeta}^2_{A}(\pi)$.

The following theorem establishes the consistency of $\hat \sigma_{1,G}^2$ for $\sigma_1^2$.  In the statement of the theorem, we make use of the following additional notation: for scalars $a$ and $b$, we define $[a \pm b] := [a - b, a +b]$, and denote by $\Phi(\cdot)$ the standard normal c.d.f.
\begin{theorem} \label{theorem:sigma1}
Under Assumptions \ref{ass:assignment} and \ref{ass:QG}, $$\hat \sigma_{1,G}^2 \stackrel{p}{\rightarrow} \sigma_1^2$$ as $G \rightarrow \infty$.  Thus, for $\sigma^2_1 > 0$ and for any $\alpha \in (0,1)$, $$P\left \{\theta_1 \in \left [\hat \theta_{1,G} \pm \frac{\hat \sigma_{1,G}}{\sqrt{G}} \Phi^{-1}\left (1 - \frac{\alpha}{2} \right ) \right ] \right \} \rightarrow 1 - 
\alpha$$ as $G \rightarrow \infty$.
\end{theorem}

\begin{remark}\label{rem:sigma_pos}
A sufficient condition under which $\sigma^2_1 > 0$ holds is that $\var[\bar{Y}_g(a) - E[\bar{Y}_g(a)|S_g]] > 0$ for some $a \in \{0, 1\}$. More generally, we expect $\sigma^2_1 > 0$ except in pathological cases such as when the distribution of outcomes is degenerate or in cases with perfect negative within-cluster correlation.
\end{remark}

\begin{remark}\label{rem:controls}
As mentioned earlier, $\hat \theta_{1,G}$ can equivalently be obtained as the estimator of the coefficient on $A_g$ in an ordinary least squares regression of $\bar Y_g$ on a constant and $A_g$. A natural next step would be to include additional baseline covariates in this linear regression.  Doing so carefully can lead to gains in efficiency;  see \cite{negi2021revisiting} and references therein for related results in the context of individual-level randomized experiments. In Appendix \ref{sec:adjust}, we describe one such adjustment strategy.  
\end{remark}

\subsection{Size-weighted Cluster-level Average Treatment Effect}\label{sec:size}
In this section, we consider the estimation of $\theta_2$ defined in \eqref{eq:SECT}.  To this end, consider the following difference-in-``weighted average of averages'' estimator: 
\begin{equation} \label{eq:theta3hat}
\hat \theta_{2,G} {}:= \frac{\sum_{1 \leq g \leq G} \bar Y_g N_g A_g}{\sum_{1 \leq g \leq G} N_g A_g} - \frac{\sum_{1 \leq g \leq G} \bar Y_g N_g (1- A_g)}{\sum_{1 \leq g \leq G} N_g (1- A_g)}~,
\end{equation}
where $\bar Y_g$ is defined as in \eqref{eq:barYg}. Note that $\hat \theta_{2,G}$ may be obtained as the estimator of the coefficient on $A_g$ in the following weighted least squares regression:
$$ \texttt{regress } Y_{i,g} \texttt{ on } \text{constant} +  A_g \texttt{ using weights } {N_g/|\mathcal{M}_g|}~.$$
In the special case where $\mathcal{M}_g = \{1, 2, \ldots, N_g\}$ for all $1 \le g \le G$ with probability one, we have $\hat \theta_{2,G} = \hat{\theta}^{\rm alt}_G$ (i.e. the weights collapse to $1$). The following theorem derives the asymptotic behavior of this estimator.

\begin{theorem} \label{thm:mainSECT}
Under Assumptions \ref{ass:assignment} and \ref{ass:QG}, $$\sqrt G (\hat \theta_{2,G} - \theta_2) \stackrel{d}{\rightarrow} N(0,\sigma_2^2)$$ as $G \rightarrow \infty$, where 
\begin{equation}
\sigma_2^2 := \frac{1}{\pi} {\text Var}[\tilde Y^{\dagger}_g(1)] + \frac{1}{1-\pi} {\text Var}[\tilde Y^{\dagger}_g(0)] + E[(\tilde{m}_1(S_g) - \tilde{m}_0(S_g))^2] + E\bigg[\tau(S_g)\left(\frac{1}{\pi}\tilde{m}_1(S_g) + \frac{1}{1-\pi}\tilde{m}_0(S_g)\right)^2\bigg]~,
\label{eq:mainSECT0}
\end{equation}
with
\begin{eqnarray*}
\tilde{Y}_g(a) &:=& \frac{N_g}{E[N_g]}\left(\bar{Y}_g(a) - \frac{E[\bar{Y}_g(a)N_g]}{E[N_g]}\right)~, \\
\tilde{Y}^{\dagger}_g(a) &:=& \tilde{Y}_g(a) - E[\tilde{Y}_g(a)|S_g]~, \\
\tilde{m}_a(S_g)&:=& E[\tilde{Y}_g(a)|S_g] - E[\tilde{Y}_g(a)]~,
\end{eqnarray*}
and $\pi$, $\tau(\cdot)$ defined as in Assumption \ref{ass:assignment}.
\end{theorem}

\begin{remark}\label{rem:side_id}
Note that, unlike $\hat{\theta}^{\rm alt}_G$ and $\hat{\theta}_{1,G}$, the estimator $\hat{\theta}_{2,G}$ cannot be computed without explicit knowledge of $N^{(G)} := (N_g : 1 \leq g \leq G)$. As explained in Section \ref{sec:parameters}, however, $\theta_2$ is in some instances equal to $\vartheta$, which may be consistently estimated using $\hat{\theta}^{\rm alt}_G$.  
\end{remark}

\begin{remark}\label{rem:finpop_contrast}
As mentioned in the introduction, our analysis is distinct from the finite-population analyses undertaken in, for instance, \cite{su2021model}. Here we compare our $\sigma_2^2$ to the variance of the difference-in-means estimator from such an analysis. Specifically, in the special case where $\mathcal{M}_g = \{1, 2, \ldots, N_g\}$ and $|\mathcal{S}| = 1$, $\sigma^2_2$ could alternatively be written as
\[\sigma_2^2 := \frac{1}{E[N_g]^2}\left(\frac{E\left[\left(\sum_{1 \le i \le N_g}\epsilon_{i,g}(1)\right)^2\right]}{\pi} + \frac{E\left[\left(\sum_{1 \le i \le N_g}\epsilon_{i,g}(0)\right)^2\right]}{1 - \pi}\right)~. \]
with
\[\epsilon_{i,g}(a) = Y_{i,g}(a) - \frac{E[N_g\bar{Y}_g(a)]}{E[N_g]}~.\]
It follows from Theorem 1 of \cite{su2021model} that the finite-population design-based variance is given by:
\begin{eqnarray*}
\sigma^2_{2,G,{\rm finpop}} &:=& \left(\frac{G}{N}\right)^2\left(
\frac{1}{G}\sum_{1 \le g \le G}\left[\frac{\left(\sum_{1\le i \le N_g}\tilde{\epsilon}_{i,g}(1)\right)^2}{\pi} + \frac{\left(\sum_{1\le i \le N_g}\tilde{\epsilon}_{i,g}(0)\right)^2}{1 - \pi}\right] \right.\\
&& \hspace{2cm} \left. - \frac{1}{G}\sum_{1 \le g \le G}\left[\sum_{1\le i \le N_g}\left(\tilde{\epsilon}_{i,g}(1) - \tilde{\epsilon}_{i,g}(0)\right)\right]^2\right)~,
\end{eqnarray*}
where
\begin{eqnarray*}
N &:=& \sum_{1 \le g \le G}N_g \\
\tilde{\epsilon}_{i,g}(a) &:=& Y_{i,g}(a) - \frac{1}{N}\sum_{1 \le g \le G} \sum_{1 \le i \le N_g}Y_{i,g}(a)~.
\end{eqnarray*}
We emphasize that in the finite-population framework adopted by \cite{su2021model}, all of the above quantities are non-random, and are derived under complete randomization with $|\mathcal{S}|=1$. From this, we see that the comparison between $\sigma^2_{2,G,{\rm finpop}}$ and $\sigma_2^2$ exactly mimics the comparison between the super-population and finite-population variance expressions for the difference-in-means estimator in the non-clustered setting \citep[see, for example,][]{ding2017bridging}. In particular, $\sigma^2_{2,G,{\rm finpop}}$ is made up of two terms: the first term corresponds to a finite-population analog of $\sigma^2_2$, whereas the second term, which enters negatively, can be interpreted as the gain in precision which results from observing the entire population. 
\end{remark}

\begin{remark}
As discussed in Remark \ref{rem:comparingestimands}, $\theta_1 = \theta_2$ whenever $N_g = k$ for all $1 \leq g \leq G$.  Furthermore, in this case we have $\hat \theta_{1,G} = \hat \theta_{2,G}$ and thus $\sigma_1^2 = \sigma_2^2$ as well.
\end{remark}

In parallel with our development in the preceding section, we note that $\hat{\theta}_{2,G}$ is most efficient when paired with an assignment mechanism that features $\tau(s) = 0$ for all $s \in \mathcal{S}$, and we show that the probability limit of the cluster-robust variance estimator is generally too large relative to $\sigma_2^2$.  
\begin{theorem}\label{thm:limCR}
Let $\tilde{\sigma}_{2,G}^2$ denote the cluster-robust estimator of the variance of the coefficient of $A_g$ in a weighted-least squares regression of $Y_{ig}$ on a constant and $A_g$, with weights equal to ${N_g/|\mathcal{M}_g|}$. This estimator can be written as
\begin{equation}
    \tilde\sigma^2_{2,G} = \tilde\sigma^2_{2,G}(1) + \tilde\sigma^2_{2,G}(0)~,
    \label{eq:hatsigma2exp}
\end{equation}
where, for $a \in \{0,1\}$, we define
\begin{equation}\label{eq:hatsigma2}
\tilde{\sigma}^2_{2,G}(a) := \frac{1}{\left(\frac{1}{G}\sum_{1\le g \le G}{N_g}I\{A_g = a\}\right)^2}\frac{1}{G}\sum_{1\le g \le G}\left[\left(\frac{N_g}{|\mathcal{M}_g|}\right)^2I\{A_g = a\}\left(\sum_{i\in \mathcal{M}_g} \hat{\epsilon}_{i,g}(a)\right)^2\right]~,
\end{equation}
with
\[\hat{\epsilon}_{i,g}(a) := Y_{i,g} - \frac{1}{\sum_{1\le g \le G}N_gI\{A_g = a\}}\sum_{1 \le g \le G}N_g\bar{Y}_gI\{A_g = a\}~.\]
Then, under Assumptions \ref{ass:assignment} and \ref{ass:QG}, 
\begin{equation}
    \tilde\sigma^2_{2,G} \xrightarrow{p} \frac{1}{\pi}{\text Var}[\tilde{Y}_g(1)] + \frac{1}{1 - \pi}{\text Var}[\tilde{Y}_g(0)] \ge \sigma^2_2~
    \label{eq:thm3p6}
\end{equation}
with equality if and only if \[(\pi(1 - \pi) - \tau(s))\left(\frac{1}{\pi}\tilde{m}_1(s) + \frac{1}{1 - \pi}\tilde{m}_0(s)\right)^2 = 0~,\] for every $s \in \mathcal{S}$.
\end{theorem}

To facilitate the use of Theorem \ref{thm:mainSECT} for inference about $\theta_2$, we now provide an estimator for $\sigma^2_2$ which is consistent. Similarly to \cite{bai2022pairs} and \cite{liu2023inference}, we construct our estimator using a feasible analog of $\tilde{Y}_g(a)$ given by
\[\hat{Y}_g :=  \frac{N_g}{\frac{1}{G}\sum_{1 \le j \le G}N_j}\left(\bar{Y}_g - \frac{\frac{1}{G}\sum_{1 \le j \le G}\bar{Y}_jI\{A_j = A_g\}N_j}{\frac{1}{G}\sum_{1 \le j \le G}I\{A_j = A_g\}N_j}\right)~.\] 
We then define the following estimators:
\begin{eqnarray}
\hat{\xi}^2_{\tilde{Y}}(\pi) &:=& \frac{1}{\pi}\left(\hat{\mu}^{\hat{Y}^2}_{G,1} - \sum_{s \in \mathcal{S}}\frac{G(s)}{G}\hat{\mu}^{\hat{Y}}_{G,1}(s)^2\right) + \frac{1}{1 - \pi}\left(\hat{\mu}^{\hat{Y}^2}_{G,0} - \sum_{s \in \mathcal{S}}\frac{G(s)}{G}\hat{\mu}^{\hat{Y}}_{G,0}(s)^2\right) \label{eq:xiY} \\
\hat{\xi}^2_H &:=&  \sum_{s \in \mathcal{S}}\frac{G(s)}{G}\left(\left(\hat{\mu}^{\hat{Y}}_{G,1}(s) - \hat{\mu}^{\hat{Y}}_{G,1}\right) - \left(\hat{\mu}^{\hat{Y}}_{G,0}(s) - \hat{\mu}^{\hat{Y}}_{G,0}\right)\right)^2 \label{eq:xiH} \\
\hat{\xi}^2_{A}(\pi) &:=& \sum_{s \in \mathcal{S}}\tau(s)\frac{G(s)}{G}\left(\frac{1}{\pi}\left(\hat{\mu}^{\hat{Y}}_{G,1}(s) - 
\hat{\mu}^{\hat{Y}}_{G,1}\right) + \frac{1}{1-\pi}\left(\hat{\mu}^{\hat{Y}}_{G,0}(s) - \hat{\mu}^{\hat{Y}}_{G,0}\right)\right)^2~, \label{eq:xiA} 
\end{eqnarray}
and set $\hat{\sigma}^2_{2,G} :=  \hat{\xi}^2_{\tilde{Y}}(\pi) + \hat{\xi}^2_H + \hat{\xi}^2_{A}(\pi)$.
The following theorem establishes the consistency of $\hat \sigma_{2,G}^2$ for $\sigma_2^2$.  In the statement of the theorem, we again make use of the notation introduced preceding Theorem \ref{theorem:sigma1}.

\begin{theorem} \label{theorem:sigma2}
Under Assumptions \ref{ass:assignment} and \ref{ass:QG}, $$\hat \sigma_{2,G}^2 \stackrel{p}{\rightarrow} \sigma_2^2$$ as $G \rightarrow \infty$.  Thus, for $\sigma_2^2 > 0$ and for any $\alpha \in (0,1)$, $$P\left \{\theta_2 \in \left [\hat \theta_{2,G} \pm \frac{\hat \sigma_{2,G}}{\sqrt{G}} \Phi^{-1}\left (1 - \frac{\alpha}{2} \right ) \right ] \right \} \rightarrow 1 - 
\alpha$$ as $G \rightarrow \infty$.
\end{theorem}

\begin{remark}
It can be shown that a sufficient condition under which $\sigma^2_2 > 0$ holds is that $\var[\tilde{Y}_g(a) - E[\tilde{Y}_g(a)|S_g]] > 0$ for some $a \in \{0,1\}$. Similarly to the discussion in Remark \ref{rem:sigma_pos}, we expect this to hold outside of pathological cases.
\end{remark}

\begin{remark} \label{rem:theta4hat} 
A natural next step to consider would be the inclusion of additional baseline covariates in the regression specification. To that end, in Appendix \ref{sec:adjust}, we consider an adjustment strategy based on an alternative estimator of $\theta_2$ given by
\begin{equation} \label{eq:theta4hat}
\hat \theta^{\rm sd}_{2,G} := \frac{\frac{1}{G} \sum_{1 \leq g \leq G} \bar Y_g N_g A_g}{\bar N_G \bar A_G} - \frac{\frac{1}{G} \sum_{1 \leq g \leq G}  \bar Y_g N_g (1- A_g)}{\bar N_G (1 - \bar A_G)}~,
\end{equation}
where $\bar N_G := \frac{1}{G} \sum_{1 \leq g \leq G} N_g$ and $\bar A_G := \frac{1}{G} \sum_{1 \leq g \leq G} A_g$.  Note that $\hat{\theta}^{\rm sd}_{2,G}$ may be obtained as the estimator of the coefficient of $A_g$ in an ordinary least squares regression of $\Gamma_{g,G} := \bar{Y}_g(N_g/\bar{N}_G)$ on a constant and $A_g$. \cite{su2021model} argue in the context of completely randomized experiments that $\hat{\theta}^{\rm sd}_{2,G}$ is well suited for the inclusion of additional baseline covariates (particularly when $N_g$ is included as a regressor). Moreover, we conjecture in the appendix that a fully non-parametric covariate adjustment strategy based on this estimator is, in fact, efficient.
\end{remark}

\section{Simulations}\label{sec:sims}
In this section, we illustrate the results in Section \ref{sec:main} with a simulation study. In all cases, data are generated as
\begin{equation}\label{eq:dgp}
  Y_{i,g}(a) = \eta_{g}(a)Z_{g,1} + \tilde m_a(Z_{g,2}) +  U_{i,g}(a)~,
\end{equation}
for $a\in\{0,1\}$, where
\begin{itemize}
  \item $\eta_{g}(a)$ are i.i.d.\ with $\eta_{g}(0)\sim U[0,1]$, and $\eta_{g}(1)\sim U[0,5]$.
  \item $U_{i,g}(a)$ are i.i.d.\ with $U_{i,g}(a)\sim N(0, \sigma(a))$ and $\sigma(1)=\sqrt{2}>\sigma(0)=1$.
  \item $\tilde m_a (Z_{g,2})= m_a(Z_{g,2}) - E[m_a(Z_{g,2})]$ where 
  $$ m_1(Z_{g,2}) = Z_{g,2}\quad \text{ and }\quad m_0(Z_{g,2}) = -\log(Z_{g,2}+3)I\{Z_{g,2}\le \frac{1}{2}\}~.$$
  \item The distribution of $Z_g:=(Z_{g,1},Z_{g,2})$ varies by design as described below. 
\end{itemize}

We consider three alternative distributions of cluster sizes. To describe them, let $BB(a,b,n_{\rm supp})$ be the Beta-Binomial distribution with parameters $a$ and $b$ and support on $\{0,\dots,n_{\rm supp}\}$. We then define $$ N_g = 10(B+1) \quad \text{ where } \quad B\sim BB(a,b,n_{\rm supp})~,$$
for the following values of $(a,b)$ and $n_{\rm supp}$, 
\begin{itemize}
  \item $(a,b)= (1,1)$: uniform pmf on $10$ to $N_{\rm max} = 10(n_{\rm supp}+1)$.
  \item $(a,b)= (0.4,0.4)$: U-shaped pmf on $10$ to $N_{\rm max} = 10(n_{\rm supp}+1)$.
  \item $(a,b)= (10,50)$: bell-shaped pmf on $10$ to $N_{\rm max} = 10(n_{\rm supp}+1)$ with a long right tail.
\end{itemize}

\begin{figure}
\begin{center}
\includegraphics[width=0.9\textwidth]{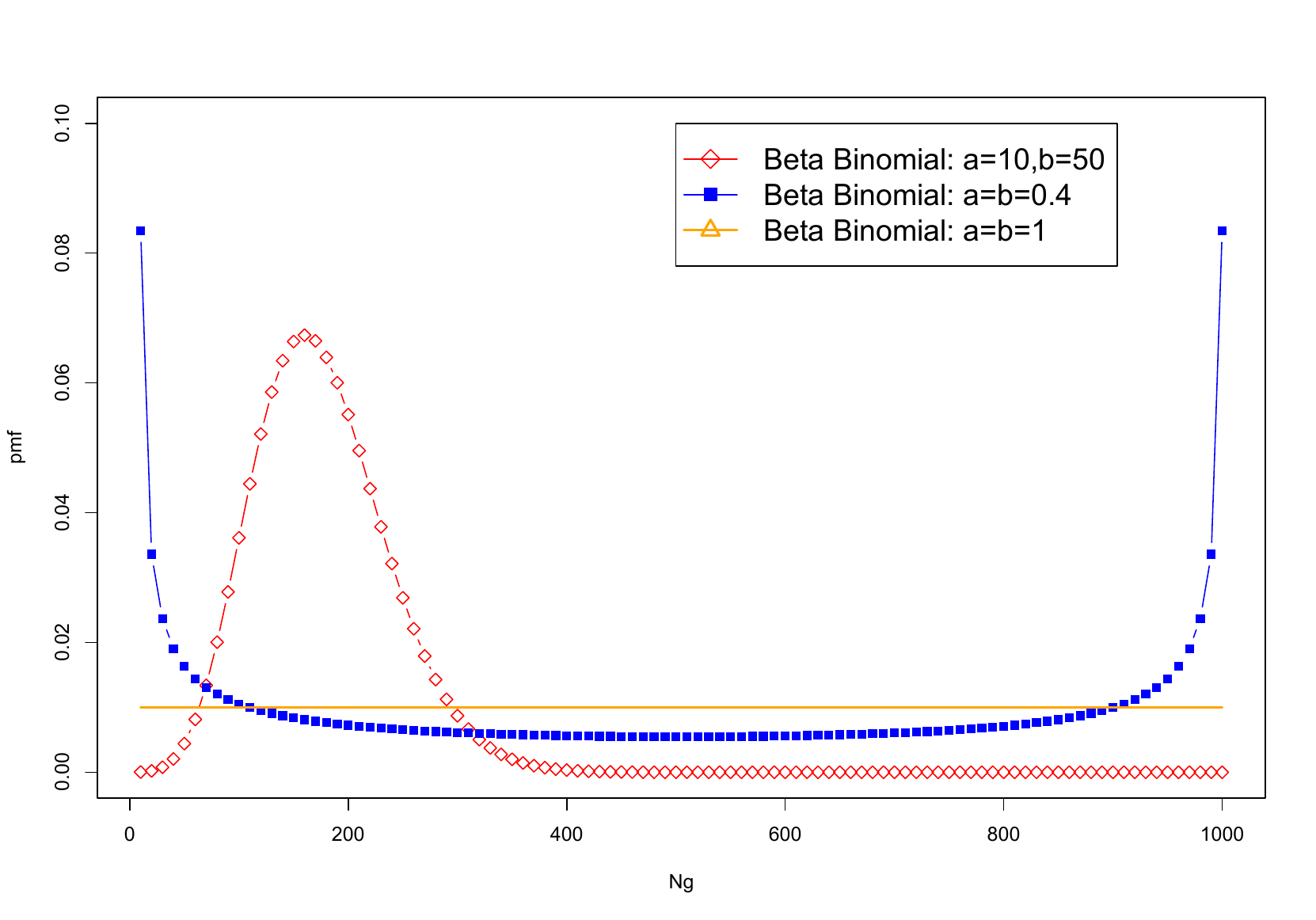}
\end{center}
\caption{Three probability mass functions (pmf) of $N_g$ when $N_{\rm max}=1000$}\label{fig:pmf}
\end{figure}

For each of the three distributions of cluster sizes, we consider three alternative ways to draw the observations within the $g$th cluster that are sampled by the researcher, $\mathcal M_g$: (a) $|\mathcal M_g|=N_g$, (b) $|\mathcal M_g|=10$, and (e) $|\mathcal M_g|=\max\{10,\min\{\gamma N_g,200\}\}$ with $\gamma=0.4$.

The combination of the three distributions of $N_g$ and the three distributions of $\mathcal M_g$ leads to 9 alternative specifications. For each of these specifications, we consider in addition two designs for the distribution of $Z_g$ as follows, 
\begin{itemize}
  \item {\bf Design 1}: $Z_{g,1}\independent N_g$ with $Z_{g,1}\in\{-1,1\}$ i.i.d.\ with $p_z\equiv P\{Z_{g,1}=1\}=1/2$. 
  \item {\bf Design 2}: $Z_{g,1} = Z_{g,\rm big}I\{N_g\ge E[N_g]\}+Z_{g,\rm small}I\{N_g <E[N_g]\}$ where $Z_{g,\rm big}\in\{-1,1\}$ with $p_z=3/4$ and $Z_{g,\rm small}\in\{-1,1\}$ i.i.d.\ with $p_z=1/4$.  
  \item In both designs, $Z_{g,2}\independent N_g$ with $Z_{g,2} \sim \text{Beta}(2,2)$ (re-centered and re-scaled by the population mean and variance to have mean zero and variance one).
\end{itemize}

Finally, treatment assignment $A_g$ follows a covariate-adaptive randomization (CAR) mechanism based on stratified block randomization with $\pi=\frac{1}{2}$ within each stratum. Concretely, we stratify the observations as follows:
\begin{itemize} 
\item {\bf CAR-1}: $S(\cdot) \independent N_g$ where strata are determined by dividing the support of $Z_{g,2}$ into $|\mathcal{S}|=10$ intervals of equal length and letting $S(Z_{g,2})$ be the function that returns the interval in which $Z_{g,2}$ lies.
\item {\bf CAR-2}: $S(\cdot)\not \independent N_g$ where strata are determined by the cartesian-product of dividing the support of $Z_{g,2}$ into $|\mathcal{S}|/2$ intervals of equal length and letting $S(Z_{g,2})$ be the function that returns the interval in which $Z_{g,2}$ lies, and dividing the support of $N_g$ by whether $N_g$ is above or below the median of $N_g$. The total number of strata is $|\mathcal{S}| = |\mathcal{S}|/2\times 2$ and we set $|\mathcal{S}|=10$ as in CAR-1. 
\end{itemize}
In principle, we could also consider other assignment mechanisms such as simple random sampling. We decided to focus on stratified block randomization because it is prevalent in practice and our results show that it dominates other mechanisms for which $\tau(s) \ne 0$ in terms of asymptotic efficiency.

The model in \eqref{eq:dgp}, as well as the two CAR designs, follow closely the original designs for covariate-adaptive randomization with individual-level data considered in \cite{bugni/canay/shaikh:2018,bugni/canay/shaikh:2019}. Note that for these designs, we obtain that
\begin{equation*}
  E[Y_{i,g}(1)-Y_{i,g}(0)|N_g] = 2E[Z_{g,1}|N_g]~.
\end{equation*}
In Design 1, it follows that $\theta_1=\theta_2=0$. In Design 2, on the other hand, it follows that 
\begin{equation*}
  E[Z_{g,1}|N_g] = \left \lbrace \begin{array}{cc}
  E[Z_{g,\rm big}]= 1/2 & \text{ if } N_g\ge E[N_g]\\
  E[Z_{g,\rm small}]= -1/2 & \text{ if } N_g< E[N_g]
  \end{array}\right. ~,
\end{equation*}
and so 
\begin{align*}
  \theta_1 &= 2P\{N_g\ge E[N_g]\} - 1\\
  \theta_2 &= E\left[\frac{N_g}{E[N_g]}\mid N_g\ge E[N_g]\right]P\{N_g\ge E[N_g]\}-E\left[\frac{N_g}{E[N_g]}\mid N_g< E[N_g]\right]P\{N_g< E[N_g]\}~.
\end{align*}

For each of the above 9 specifications and for each design of $Z_{g,1}$ and treatment assignment mechanism, we generate samples using $G \in \{100, 5000\}$ and report: the true values of $(\theta_1,\theta_2)$ defined in \eqref{eq:ECA} and \eqref{eq:SECT}, the average across simulations of the estimated values $(\hat \theta_{1,G},\hat\theta_{2,G})$ defined in \eqref{eq:theta2hat} and \eqref{eq:theta3hat}, the average across simulations of the estimated standard deviations $(\hat \sigma_{1,G}, \hat \sigma_{2,G})$ defined in Sections \ref{sec:equal} and \ref{sec:size}, and the empirical coverage of the $95\%$ confidence intervals defined in Theorems \ref{theorem:sigma1} and \ref{theorem:sigma2}. The results of our simulations are presented in Tables \ref{table:1} to \ref{table:6}, where in all cases we conducted $5,000$ replications. In each table, we find that our empirical coverage probabilities are close to $95\%$ in all cases.

\begin{table}[h!]
\centering
\small
\begin{tabular}{ll|rrrrrrrr}
\hline
 \multicolumn{2}{l}{ CAR-1 $\mid$ Design 1} & \multicolumn{2}{c}{True values} & \multicolumn{2}{c}{Estimated} & \multicolumn{2}{c}{Estimated s.d.} & \multicolumn{2}{c}{Cov. Prob.}   \\ 
 \cmidrule(lr){3-4} \cmidrule(lr){5-6} \cmidrule(lr){7-8} \cmidrule(lr){9-10}
 \multicolumn{1}{l}{$\mathcal M_g$}  & \multicolumn{1}{l}{$N_g$} & \multicolumn{1}{c}{$\theta_1$} & \multicolumn{1}{c}{$\theta_2$} & \multicolumn{1}{c}{$\hat\theta_{1,G}$} & \multicolumn{1}{c}{$\hat\theta_{2,G}$} & \multicolumn{1}{c}{$\hat\sigma_{1,G}$} & \multicolumn{1}{c}{$\hat\sigma_{2,G}$} & \multicolumn{1}{c}{$CS_{1,G}$} & \multicolumn{1}{c}{$CS_{2,G}$} \\ \hline\hline
$N_g$& $Bb(1,1)$ & 0.0000 & 0.0000 & -0.0016 & 0.0016 & 4.2885 & 4.9375 & 0.9440 & 0.9426 \\ 
     & $Bb(0.4,0.4)$ & 0.0000 & 0.0000 & -0.0080 & -0.0070 & 4.2864 & 5.2952 & 0.9454 & 0.9310 \\ 
     & $Bb(10,50)$ & 0.0000 & 0.0000 & -0.0001 & -0.0010 & 4.2808 & 4.5852 & 0.9444 & 0.9486 \\ 
\hline
$10$ & $Bb(1,1)$ & 0.0000 & 0.0000 & 0.0007 & -0.0000 & 4.3297 & 4.9780 & 0.9426 & 0.9330 \\ 
     & $Bb(0.4,0.4)$ & 0.0000 & 0.0000 & -0.0008 & 0.0015 & 4.3385 & 5.3582 & 0.9460 & 0.9438 \\ 
     & $Bb(10,50)$ & 0.0000 & 0.0000 & 0.0023 & 0.0063 & 4.3414 & 4.6545 & 0.9436 & 0.9440 \\
\hline
$\gamma N_g$ & $Bb(1,1)$ & 0.0000 & 0.0000 & 0.0090 & 0.0087 & 4.2827 & 4.9091 & 0.9346 & 0.9338 \\ 
    & $Bb(0.4,0.4)$ & 0.0000 & 0.0000 & -0.0089 & -0.0044 & 4.2871 & 5.3089 & 0.9426 & 0.9376 \\ 
    & $Bb(10,50)$ & 0.0000 & 0.0000 & 0.0049 & 0.0064 & 4.2994 & 4.5949 & 0.9436 & 0.9470 \\ 
   \hline
\multicolumn{2}{l}{CAR-1 $\mid$ Design 2} & \multicolumn{1}{c}{$\theta_1$} & \multicolumn{1}{c}{$\theta_2$} & \multicolumn{1}{c}{$\hat\theta_{1,G}$} & \multicolumn{1}{c}{$\hat\theta_{2,G}$} & \multicolumn{1}{c}{$\hat\sigma_{1,G}$} & \multicolumn{1}{c}{$\hat\sigma_{2,G}$} & \multicolumn{1}{c}{$CS_{1,G}$} & \multicolumn{1}{c}{$CS_{2,G}$} \\
   \hline
$N_g$ & $Bb(1,1)$ & 0.0000 & 0.4900 & -0.0036 & 0.4819 & 4.2792 & 4.7416 & 0.9474 & 0.9454 \\ 
   & $Bb(0.4,0.4)$ & 0.0000 & 0.6581 & -0.0142 & 0.6387 & 4.2870 & 5.0338 & 0.9458 & 0.9424 \\ 
    & $Bb(10,50)$ & -0.1407 & 0.1625 & -0.0822 & 0.2172 & 4.2825 & 4.5250 & 0.9342 & 0.9440 \\ 
\hline
  $10$ & $Bb(1,1)$ & 0.0000 & 0.4900 & -0.0044 & 0.4844 & 4.3439 & 4.8198 & 0.9448 & 0.9454 \\ 
       & $Bb(0.4,0.4)$ & 0.0000 & 0.6581 & -0.0168 & 0.6366 & 4.3399 & 5.1226 & 0.9402 & 0.9426 \\
        & $Bb(10,50)$ & -0.1407 & 0.1625 & -0.0807 & 0.2173 & 4.3346 & 4.5764 & 0.9480 & 0.9526 \\ 
\hline
  $\gamma N_g$ & $Bb(1,1)$ & 0.0000 & 0.4900 & -0.0129 & 0.4673 & 4.2884 & 4.7446 & 0.9548 & 0.9428 \\ 
          & $Bb(0.4,0.4)$ & 0.0000 & 0.6581 & -0.0147 & 0.6366 & 4.2852 & 5.0404 & 0.9404 & 0.9462 \\ 
           & $Bb(10,50)$ & -0.1407 & 0.1625 & -0.0800 & 0.2201 & 4.2922 & 4.5273 & 0.9416 & 0.9398 \\ 
\hline\hline
\end{tabular}
\caption {\label{tab:table1.g100.n500} \small Results for $G=100$, $N_{\rm max}=500$, $Z_g \independent N_g$ (Design 1), $Z_g \not\independent N_g$ (Design 2), and CAR-1.} \label{table:1}
\end{table}

\vspace{-3mm}

\begin{table}[h!]
\centering
\small
\begin{tabular}{ll|rrrrrrrr}
\hline
 \multicolumn{2}{l}{CAR-1 $\mid$ Design 1} & \multicolumn{2}{c}{True values} & \multicolumn{2}{c}{Estimated} & \multicolumn{2}{c}{Estimated s.d.} & \multicolumn{2}{c}{Cov. Prob.}   \\ 
 \cmidrule(lr){3-4} \cmidrule(lr){5-6} \cmidrule(lr){7-8} \cmidrule(lr){9-10}
 \multicolumn{1}{l}{$\mathcal M_g$}  & \multicolumn{1}{l}{$N_g$} & \multicolumn{1}{c}{$\theta_1$} & \multicolumn{1}{c}{$\theta_2$} & \multicolumn{1}{c}{$\hat\theta_{1,G}$} & \multicolumn{1}{c}{$\hat\theta_{2,G}$} & \multicolumn{1}{c}{$\hat\sigma_{1,G}$} & \multicolumn{1}{c}{$\hat\sigma_{2,G}$} & \multicolumn{1}{c}{$CS_{1,G}$} & \multicolumn{1}{c}{$CS_{2,G}$} \\ \hline\hline
$N_g$ & $Bb(1,1)$ & 0.0000 & 0.0000 & -0.0004 & -0.0057 & 4.2824 & 4.9264 & 0.9408 & 0.9382 \\ 
   & $Bb(0.4,0.4)$ & 0.0000 & 0.0000 & 0.0002 & 0.0008 & 4.2835 & 5.3107 & 0.9388 & 0.9388 \\ 
    & $Bb(10,50)$ & 0.0000 & 0.0000 & -0.0061 & -0.0080 & 4.2803 & 4.5365 & 0.9438 & 0.9436 \\ 
\hline 
  $10$ & $Bb(1,1)$ & 0.0000 & 0.0000 & -0.0075 & -0.0102 & 4.3407 & 4.9950 & 0.9374 & 0.9412 \\ 
       & $Bb(0.4,0.4)$ & 0.0000 & 0.0000 & 0.0009 & 0.0052 & 4.3405 & 5.3903 & 0.9386 & 0.9426 \\ 
        & $Bb(10,50)$ & 0.0000 & 0.0000 & -0.0129 & -0.0147 & 4.3358 & 4.5938 & 0.9442 & 0.9490 \\ 
 \hline
  $\gamma N_g$ & $Bb(1,1)$ & 0.0000 & 0.0000 & -0.0115 & -0.0146 & 4.2837 & 4.9416 & 0.9438 & 0.9442 \\ 
          & $Bb(0.4,0.4)$ & 0.0000 & 0.0000 & -0.0095 & -0.0077 & 4.2907 & 5.3277 & 0.9386 & 0.9392 \\ 
           & $Bb(10,50)$ & 0.0000 & 0.0000 & -0.0022 & -0.0051 & 4.2841 & 4.5295 & 0.9434 & 0.9446 \\ 
   \hline
\multicolumn{2}{l}{CAR-1 $\mid$ Design 2} & \multicolumn{1}{c}{$\theta_1$} & \multicolumn{1}{c}{$\theta_2$} & \multicolumn{1}{c}{$\hat\theta_{1,G}$} & \multicolumn{1}{c}{$\hat\theta_{2,G}$} & \multicolumn{1}{c}{$\hat\sigma_{1,G}$} & \multicolumn{1}{c}{$\hat\sigma_{2,G}$} & \multicolumn{1}{c}{$CS_{1,G}$} & \multicolumn{1}{c}{$CS_{2,G}$} \\
   \hline
$N_g$ & $Bb(1,1)$ & 0.0000 & 0.4950 & 0.0005 & 0.4895 & 4.2746 & 4.7550 & 0.9396 & 0.9410 \\ 
   & $Bb(0.4,0.4)$ & 0.0000 & 0.6690 & 0.0039 & 0.6711 & 4.2855 & 5.0689 & 0.9480 & 0.9510 \\ 
    & $Bb(10,50)$ & -0.0635 & 0.2100 & -0.0691 & 0.2024 & 4.2806 & 4.4677 & 0.9414 & 0.9486 \\ 
\hline
  $10$ & $Bb(1,1)$ & 0.0000 & 0.4950 & -0.0071 & 0.4834 & 4.3384 & 4.8310 & 0.9456 & 0.9442 \\ 
       & $Bb(0.4,0.4)$ & 0.0000 & 0.6690 & -0.0054 & 0.6533 & 4.3341 & 5.1559 & 0.9456 & 0.9446 \\ 
        & $Bb(10,50)$ & -0.0635 & 0.2100 & -0.0687 & 0.1973 & 4.3408 & 4.5444 & 0.9426 & 0.9436 \\ 
\hline
  $\gamma N_g$ & $Bb(1,1)$ & 0.0000 & 0.4950 & -0.0192 & 0.4723 & 4.2826 & 4.7515 & 0.9456 & 0.9440 \\ 
          & $Bb(0.4,0.4)$ & 0.0000 & 0.6690 & -0.0060 & 0.6638 & 4.2876 & 5.0702 & 0.9432 & 0.9460 \\ 
           & $Bb(10,50)$ & -0.0635 & 0.2100 & -0.0708 & 0.1987 & 4.2862 & 4.4822 & 0.9446 & 0.9472 \\ 
\hline\hline
\end{tabular}
\caption {\label{tab:table2.g100.n1000} \small Results for $G=100$, $N_{\rm max}=1,000$, $Z_g \independent N_g$ (Design 1), $Z_g \not\independent N_g$ (Design 2), and CAR-1.}  
\end{table}

\pagebreak
\begin{table}[h!]
\centering
\small
\begin{tabular}{ll|rrrrrrrr}
\hline
 \multicolumn{2}{l}{CAR-1 $\mid$ Design 1} & \multicolumn{2}{c}{True values} & \multicolumn{2}{c}{Estimated} & \multicolumn{2}{c}{Estimated s.d.} & \multicolumn{2}{c}{Cov. Prob.}   \\ 
 \cmidrule(lr){3-4} \cmidrule(lr){5-6} \cmidrule(lr){7-8} \cmidrule(lr){9-10}
 \multicolumn{1}{l}{$\mathcal M_g$}  & \multicolumn{1}{l}{$N_g$} & \multicolumn{1}{c}{$\theta_1$} & \multicolumn{1}{c}{$\theta_2$} & \multicolumn{1}{c}{$\hat\theta_{1,G}$} & \multicolumn{1}{c}{$\hat\theta_{2,G}$} & \multicolumn{1}{c}{$\hat\sigma_{1,G}$} & \multicolumn{1}{c}{$\hat\sigma_{2,G}$} & \multicolumn{1}{c}{$CS_{1,G}$} & \multicolumn{1}{c}{$CS_{2,G}$} \\ \hline\hline
$N_g$ & $Bb(1,1)$ & 0.0000 & 0.0000 & 0.0003 & -0.0001 & 4.3688 & 5.0513 & 0.9480 & 0.9546 \\ 
   & $Bb(0.4,0.4)$ & 0.0000 & 0.0000 & -0.0009 & -0.0010 & 4.3740 & 5.4474 & 0.9494 & 0.9530 \\ 
    & $Bb(10,50)$ & 0.0000 & 0.0000 & 0.0008 & 0.0011 & 4.3690 & 4.6261 & 0.9532 & 0.9492 \\ 
\hline
  $10$ & $Bb(1,1)$ & 0.0000 & 0.0000 & 0.0001 & -0.0003 & 4.4315 & 5.1258 & 0.9542 & 0.9560 \\ 
       & $Bb(0.4,0.4)$ & 0.0000 & 0.0000 & 0.0006 & 0.0013 & 4.4333 & 5.5313 & 0.9542 & 0.9510 \\ 
        & $Bb(10,50)$ & 0.0000 & 0.0000 & 0.0015 & 0.0014 & 4.4330 & 4.6940 & 0.9538 & 0.9546 \\ 
\hline
  $\gamma N_g$ & $Bb(1,1)$ & 0.0000 & 0.0000 & 0.0018 & 0.0024 & 4.3725 & 5.0541 & 0.9500 & 0.9468 \\ 
          & $Bb(0.4,0.4)$ & 0.0000 & 0.0000 & -0.0008 & -0.0006 & 4.3785 & 5.4505 & 0.9548 & 0.9522 \\ 
           & $Bb(10,50)$ & 0.0000 & 0.0000 & 0.0025 & 0.0022 & 4.3766 & 4.6319 & 0.9580 & 0.9598 \\ 
   \hline
\multicolumn{2}{l}{CAR-1 $\mid$ Design 2} & \multicolumn{1}{c}{$\theta_1$} & \multicolumn{1}{c}{$\theta_2$} & \multicolumn{1}{c}{$\hat\theta_{1,G}$} & \multicolumn{1}{c}{$\hat\theta_{2,G}$} & \multicolumn{1}{c}{$\hat\sigma_{1,G}$} & \multicolumn{1}{c}{$\hat\sigma_{2,G}$} & \multicolumn{1}{c}{$CS_{1,G}$} & \multicolumn{1}{c}{$CS_{2,G}$} \\
   \hline
$N_g$ & $Bb(1,1)$ & 0.0000 & 0.4950 & 0.0019 & 0.4964 & 4.3680 & 4.8506 & 0.9598 & 0.9606 \\ 
   & $Bb(0.4,0.4)$ & 0.0000 & 0.6690 & -0.0002 & 0.6675 & 4.3750 & 5.1796 & 0.9548 & 0.9552 \\ 
    & $Bb(10,50)$ & -0.0635 & 0.2100 & -0.0642 & 0.2088 & 4.3677 & 4.5609 & 0.9492 & 0.9508 \\ 
\hline
  $10$ & $Bb(1,1)$ & 0.0000 & 0.4950 & 0.0009 & 0.4961 & 4.4332 & 4.9315 & 0.9532 & 0.9562 \\ 
       & $Bb(0.4,0.4)$ & 0.0000 & 0.6690 & -0.0006 & 0.6689 & 4.4332 & 5.2654 & 0.9586 & 0.9588 \\ 
        & $Bb(10,50)$ & -0.0635 & 0.2100 & -0.0627 & 0.2104 & 4.4299 & 4.6277 & 0.9582 & 0.9592 \\ 
\hline
  $\gamma N_g$ & $Bb(1,1)$ & 0.0000 & 0.4950 & 0.0002 & 0.4956 & 4.3732 & 4.8567 & 0.9586 & 0.9578 \\ 
          & $Bb(0.4,0.4)$ & 0.0000 & 0.6690 & -0.0000 & 0.6695 & 4.3784 & 5.1788 & 0.9586 & 0.9614 \\ 
           & $Bb(10,50)$ & -0.0635 & 0.2100 & -0.0635 & 0.2094 & 4.3744 & 4.5676 & 0.9538 & 0.9576 \\ 
\hline\hline
\end{tabular}
\caption {\label{tab:table3.g5000.n1000} \small Results for $G=5,000$, $N_{\rm max}=1,000$, $Z_g \independent N_g$ (Design 1), $Z_g \not\independent N_g$ (Design 2), and CAR-1.}  
\end{table}

\begin{table}[h!]
\centering
\small
\begin{tabular}{ll|rrrrrrrr}
\hline
 \multicolumn{2}{l}{CAR-2 $\mid$ Design 1} & \multicolumn{2}{c}{True values} & \multicolumn{2}{c}{Estimated} & \multicolumn{2}{c}{Estimated s.d.} & \multicolumn{2}{c}{Cov. Prob.}   \\ 
 \cmidrule(lr){3-4} \cmidrule(lr){5-6} \cmidrule(lr){7-8} \cmidrule(lr){9-10}
 \multicolumn{1}{l}{$\mathcal M_g$}  & \multicolumn{1}{l}{$N_g$} & \multicolumn{1}{c}{$\theta_1$} & \multicolumn{1}{c}{$\theta_2$} & \multicolumn{1}{c}{$\hat\theta_{1,G}$} & \multicolumn{1}{c}{$\hat\theta_{2,G}$} & \multicolumn{1}{c}{$\hat\sigma_{1,G}$} & \multicolumn{1}{c}{$\hat\sigma_{2,G}$} & \multicolumn{1}{c}{$CS_{1,G}$} & \multicolumn{1}{c}{$CS_{2,G}$} \\ \hline\hline
$N_g$ & $Bb(1,1)$ & 0.0000 & 0.0000 & 0.0043 & -0.0040 & 4.2769 & 4.9142 & 0.9422 & 0.9450 \\ 
   & $Bb(0.4,0.4)$ & 0.0000 & 0.0000 & -0.0001 & -0.0029 & 4.2760 & 5.2566 & 0.9366 & 0.9372 \\ 
    & $Bb(10,50)$ & 0.0000 & 0.0000 & -0.0013 & -0.0021 & 4.2917 & 4.5993 & 0.9474 & 0.9486 \\ 
\hline
  $10$ & $Bb(1,1)$ & 0.0000 & 0.0000 & -0.0073 & -0.0031 & 4.3318 & 4.9799 & 0.9456 & 0.9428 \\ 
       & $Bb(0.4,0.4)$ & 0.0000 & 0.0000 & -0.0071 & -0.0118 & 4.3325 & 5.3390 & 0.9384 & 0.9408 \\ 
        & $Bb(10,50)$ & 0.0000 & 0.0000 & 0.0031 & 0.0063 & 4.3465 & 4.6611 & 0.9334 & 0.9432 \\ 
\hline
  $\gamma N_g$ & $Bb(1,1)$ & 0.0000 & 0.0000 & -0.0061 & -0.0097 & 4.2783 & 4.9082 & 0.9460 & 0.9430 \\ 
          & $Bb(0.4,0.4)$ & 0.0000 & 0.0000 & -0.0006 & -0.0058 & 4.2903 & 5.2822 & 0.9380 & 0.9468 \\ 
           & $Bb(10,50)$ & 0.0000 & 0.0000 & -0.0078 & -0.0102 & 4.3018 & 4.6021 & 0.9426 & 0.9458 \\ 
   \hline
\multicolumn{2}{l}{CAR-2 $\mid$ Design 2} & \multicolumn{1}{c}{$\theta_1$} & \multicolumn{1}{c}{$\theta_2$} & \multicolumn{1}{c}{$\hat\theta_{1,G}$} & \multicolumn{1}{c}{$\hat\theta_{2,G}$} & \multicolumn{1}{c}{$\hat\sigma_{1,G}$} & \multicolumn{1}{c}{$\hat\sigma_{2,G}$} & \multicolumn{1}{c}{$CS_{1,G}$} & \multicolumn{1}{c}{$CS_{2,G}$} \\
   \hline
$N_g$ & $Bb(1,1)$ & 0.0000 & 0.4900 & -0.0027 & 0.4915 & 4.1664 & 4.6658 & 0.9540 & 0.9476 \\ 
   & $Bb(0.4,0.4)$ & 0.0000 & 0.6581 & -0.0029 & 0.6479 & 4.1693 & 4.9832 & 0.9518 & 0.9498 \\ 
    & $Bb(10,50)$ & -0.1407 & 0.1625 & -0.0897 & 0.2121 & 4.1669 & 4.4264 & 0.9462 & 0.9484 \\ 
\hline
  $10$ & $Bb(1,1)$ & 0.0000 & 0.4900 & -0.0015 & 0.4917 & 4.2266 & 4.7462 & 0.9516 & 0.9488 \\ 
       & $Bb(0.4,0.4)$ & 0.0000 & 0.6581 & -0.0006 & 0.6563 & 4.2239 & 5.0699 & 0.9544 & 0.9518 \\ 
        & $Bb(10,50)$ & -0.1407 & 0.1625 & -0.0795 & 0.2174 & 4.2293 & 4.4920 & 0.9464 & 0.9490 \\ 
\hline
  $\gamma N_g$ & $Bb(1,1)$ & 0.0000 & 0.4900 & -0.0035 & 0.4870 & 4.1682 & 4.6798 & 0.9496 & 0.9484 \\ 
          & $Bb(0.4,0.4)$ & 0.0000 & 0.6581 & -0.0036 & 0.6532 & 4.1737 & 4.9868 & 0.9578 & 0.9524 \\ 
           & $Bb(10,50)$ & -0.1407 & 0.1625 & -0.0877 & 0.2136 & 4.1861 & 4.4456 & 0.9516 & 0.9504 \\ 
\hline\hline
\end{tabular}
\caption {\label{tab:table4.g100.n500} \small Results for $G=100$, $N_{\rm max}=500$, $Z_g \independent N_g$ (Design 1), $Z_g \not\independent N_g$ (Design 2), and CAR-2.}  
\end{table}

\pagebreak
\begin{table}[h!]
\centering
\small
\begin{tabular}{ll|rrrrrrrr}
\hline
 \multicolumn{2}{l}{CAR-2 $\mid$ Design 1} & \multicolumn{2}{c}{True values} & \multicolumn{2}{c}{Estimated} & \multicolumn{2}{c}{Estimated s.d.} & \multicolumn{2}{c}{Cov. Prob.}   \\ 
 \cmidrule(lr){3-4} \cmidrule(lr){5-6} \cmidrule(lr){7-8} \cmidrule(lr){9-10}
 \multicolumn{1}{l}{$\mathcal M_g$}  & \multicolumn{1}{l}{$N_g$} & \multicolumn{1}{c}{$\theta_1$} & \multicolumn{1}{c}{$\theta_2$} & \multicolumn{1}{c}{$\hat\theta_{1,G}$} & \multicolumn{1}{c}{$\hat\theta_{2,G}$} & \multicolumn{1}{c}{$\hat\sigma_{1,G}$} & \multicolumn{1}{c}{$\hat\sigma_{2,G}$} & \multicolumn{1}{c}{$CS_{1,G}$} & \multicolumn{1}{c}{$CS_{2,G}$} \\ \hline\hline
$N_g$ & $Bb(1,1)$ & 0.0000 & 0.0000 & -0.0083 & -0.0120 & 4.2775 & 4.9251 & 0.9394 & 0.9402 \\ 
   & $Bb(0.4,0.4)$ & 0.0000 & 0.0000 & -0.0143 & -0.0153 & 4.2801 & 5.2966 & 0.9434 & 0.9384 \\ 
    & $Bb(10,50)$ & 0.0000 & 0.0000 & -0.0086 & -0.0115 & 4.2897 & 4.5346 & 0.9422 & 0.9450 \\ 
\hline
  $10$ & $Bb(1,1)$ & 0.0000 & 0.0000 & -0.0152 & -0.0175 & 4.3419 & 4.9979 & 0.9424 & 0.9426 \\ 
       & $Bb(0.4,0.4)$ & 0.0000 & 0.0000 & -0.0058 & -0.0066 & 4.3391 & 5.3822 & 0.9428 & 0.9410 \\ 
        & $Bb(10,50)$ & 0.0000 & 0.0000 & -0.0008 & -0.0005 & 4.3482 & 4.6036 & 0.9438 & 0.9480 \\ 
\hline
  $\gamma N_g$ & $Bb(1,1)$ & 0.0000 & 0.0000 & -0.0023 & -0.0033 & 4.2840 & 4.9329 & 0.9440 & 0.9396 \\ 
          & $Bb(0.4,0.4)$ & 0.0000 & 0.0000 & 0.0044 & 0.0003 & 4.2906 & 5.3131 & 0.9408 & 0.9424 \\ 
           & $Bb(10,50)$ & 0.0000 & 0.0000 & -0.0142 & -0.0118 & 4.2884 & 4.5387 & 0.9418 & 0.9404 \\ 
   \hline
\multicolumn{2}{l}{CAR-2 $\mid$ Design 2} & \multicolumn{1}{c}{$\theta_1$} & \multicolumn{1}{c}{$\theta_2$} & \multicolumn{1}{c}{$\hat\theta_{1,G}$} & \multicolumn{1}{c}{$\hat\theta_{2,G}$} & \multicolumn{1}{c}{$\hat\sigma_{1,G}$} & \multicolumn{1}{c}{$\hat\sigma_{2,G}$} & \multicolumn{1}{c}{$CS_{1,G}$} & \multicolumn{1}{c}{$CS_{2,G}$} \\
   \hline
$N_g$ & $Bb(1,1)$ & 0.0000 & 0.4950 & -0.0046 & 0.4938 & 4.1610 & 4.6755 & 0.9546 & 0.9514 \\ 
   & $Bb(0.4,0.4)$ & 0.0000 & 0.6690 & -0.0091 & 0.6566 & 4.1644 & 4.9964 & 0.9506 & 0.9452 \\ 
    & $Bb(10,50)$ & -0.0635 & 0.2100 & -0.0715 & 0.1989 & 4.1593 & 4.3726 & 0.9474 & 0.9498 \\ 
\hline
  $10$ & $Bb(1,1)$ & 0.0000 & 0.4950 & -0.0069 & 0.4893 & 4.2275 & 4.7582 & 0.9514 & 0.9482 \\ 
       & $Bb(0.4,0.4)$ & 0.0000 & 0.6690 & -0.0063 & 0.6612 & 4.2194 & 5.0787 & 0.9584 & 0.9520 \\ 
        & $Bb(10,50)$ & -0.0635 & 0.2100 & -0.0746 & 0.1989 & 4.2219 & 4.4408 & 0.9524 & 0.9580 \\ 
\hline
  $\gamma N_g$ & $Bb(1,1)$ & 0.0000 & 0.4950 & -0.0105 & 0.4836 & 4.1551 & 4.6726 & 0.9540 & 0.9492 \\ 
          & $Bb(0.4,0.4)$ & 0.0000 & 0.6690 & -0.0062 & 0.6631 & 4.1645 & 4.9924 & 0.9530 & 0.9512 \\ 
           & $Bb(10,50)$ & -0.0635 & 0.2100 & -0.0640 & 0.2074 & 4.1660 & 4.3863 & 0.9446 & 0.9506 \\ 
\hline\hline
\end{tabular}
\caption {\label{tab:table5.g100.n1000} \small Results for $G=100$, $N_{\rm max}=1,000$, $Z_g \independent N_g$ (Design 1), $Z_g \not\independent N_g$ (Design 2), and CAR-2.}  
\end{table}

\begin{table}[h!]
\centering
\small
\begin{tabular}{ll|rrrrrrrr}
\hline
 \multicolumn{2}{l}{CAR-2 $\mid$ Design 1} & \multicolumn{2}{c}{True values} & \multicolumn{2}{c}{Estimated} & \multicolumn{2}{c}{Estimated s.d.} & \multicolumn{2}{c}{Cov. Prob.}   \\ 
 \cmidrule(lr){3-4} \cmidrule(lr){5-6} \cmidrule(lr){7-8} \cmidrule(lr){9-10}
 \multicolumn{1}{l}{$\mathcal M_g$}  & \multicolumn{1}{l}{$N_g$} & \multicolumn{1}{c}{$\theta_1$} & \multicolumn{1}{c}{$\theta_2$} & \multicolumn{1}{c}{$\hat\theta_{1,G}$} & \multicolumn{1}{c}{$\hat\theta_{2,G}$} & \multicolumn{1}{c}{$\hat\sigma_{1,G}$} & \multicolumn{1}{c}{$\hat\sigma_{2,G}$} & \multicolumn{1}{c}{$CS_{1,G}$} & \multicolumn{1}{c}{$CS_{2,G}$} \\ \hline\hline
$N_g$ & $Bb(1,1)$ & 0.0000 & 0.0000 & 0.0007 & -0.0006 & 4.3737 & 5.0442 & 0.9570 & 0.9582 \\ 
   & $Bb(0.4,0.4)$ & 0.0000 & 0.0000 & 0.0029 & 0.0033 & 4.3788 & 5.4297 & 0.9548 & 0.9568 \\ 
    & $Bb(10,50)$ & 0.0000 & 0.0000 & 0.0001 & -0.0002 & 4.3801 & 4.6374 & 0.9500 & 0.9502 \\ 
\hline
  $10$ & $Bb(1,1)$ & 0.0000 & 0.0000 & 0.0009 & 0.0004 & 4.4386 & 5.1204 & 0.9548 & 0.9520 \\ 
       & $Bb(0.4,0.4)$ & 0.0000 & 0.0000 & 0.0015 & 0.0026 & 4.4376 & 5.5128 & 0.9562 & 0.9550 \\ 
        & $Bb(10,50)$ & 0.0000 & 0.0000 & 0.0007 & 0.0005 & 4.4435 & 4.7054 & 0.9550 & 0.9536 \\ 
\hline
  $\gamma N_g$ & $Bb(1,1)$ & 0.0000 & 0.0000 & 0.0019 & 0.0020 & 4.3780 & 5.0463 & 0.9564 & 0.9540 \\ 
          & $Bb(0.4,0.4)$ & 0.0000 & 0.0000 & 0.0006 & -0.0002 & 4.3836 & 5.4340 & 0.9610 & 0.9566 \\ 
           & $Bb(10,50)$ & 0.0000 & 0.0000 & 0.0005 & 0.0004 & 4.3865 & 4.6416 & 0.9496 & 0.9518 \\ 
   \hline
\multicolumn{2}{l}{CAR-2 $\mid$ Design 2} & \multicolumn{1}{c}{$\theta_1$} & \multicolumn{1}{c}{$\theta_2$} & \multicolumn{1}{c}{$\hat\theta_{1,G}$} & \multicolumn{1}{c}{$\hat\theta_{2,G}$} & \multicolumn{1}{c}{$\hat\sigma_{1,G}$} & \multicolumn{1}{c}{$\hat\sigma_{2,G}$} & \multicolumn{1}{c}{$CS_{1,G}$} & \multicolumn{1}{c}{$CS_{2,G}$} \\
   \hline
$N_g$ & $Bb(1,1)$ & 0.0000 & 0.4950 & 0.0011 & 0.4960 & 4.2082 & 4.7586 & 0.9644 & 0.9580 \\ 
   & $Bb(0.4,0.4)$ & 0.0000 & 0.6690 & 0.0001 & 0.6686 & 4.2122 & 5.0922 & 0.9680 & 0.9598 \\ 
    & $Bb(10,50)$ & -0.0635 & 0.2100 & -0.0638 & 0.2100 & 4.2169 & 4.4390 & 0.9624 & 0.9588 \\ 
\hline
  $10$ & $Bb(1,1)$ & 0.0000 & 0.4950 & 0.0003 & 0.4959 & 4.2745 & 4.8394 & 0.9592 & 0.9582 \\ 
       & $Bb(0.4,0.4)$ & 0.0000 & 0.6690 & 0.0006 & 0.6692 & 4.2742 & 5.1821 & 0.9648 & 0.9588 \\ 
        & $Bb(10,50)$ & -0.0635 & 0.2100 & -0.0633 & 0.2097 & 4.2829 & 4.5094 & 0.9576 & 0.9614 \\ 
\hline
  $\gamma N_g$ & $Bb(1,1)$ & 0.0000 & 0.4950 & 0.0010 & 0.4965 & 4.2125 & 4.7615 & 0.9622 & 0.9572 \\ 
          & $Bb(0.4,0.4)$ & 0.0000 & 0.6690 & 0.0005 & 0.6687 & 4.2173 & 5.0962 & 0.9622 & 0.9618 \\ 
           & $Bb(10,50)$ & -0.0635 & 0.2100 & -0.0622 & 0.2111 & 4.2239 & 4.4454 & 0.9604 & 0.9586 \\ 
\hline\hline
\end{tabular}
\caption {\label{tab:table6.g5000.n1000} \small Results for $G=5,000$, $N_{\rm max}=1,000$, $Z_g \independent N_g$ (Design 1), $Z_g \not\independent N_g$ (Design 2), and CAR-2.}  \label{table:6}
\end{table}
\clearpage
\pagebreak

\section{Empirical illustration}\label{sec:application}

In this section, we illustrate our findings by revisiting the empirical application in \cite{celhay/gertler/giovagnoli/vermeersch:2019}. These authors use a field experiment in Argentina to study the effects of temporary incentives for medical care providers to adopt early initiation of prenatal care. 
The medical literature has long recognized the benefits of early initiation of prenatal care. In particular, it allows doctors to detect and treat critical medical conditions, as well as advise mothers on proper nutrition and risk prevention activities in the period in which the fetus is most at risk. 

The field experiment in \cite{celhay/gertler/giovagnoli/vermeersch:2019} took place in Misiones, Argentina, one of the poorest provinces in the country, and with relatively high rate of maternal and child mortality. As part of the national {\it Plan Nacer} program, the Argentinean government transfers funds to medical care providers in exchange for their patient services. The study selected 37 public primary care facilities (accounting for 70\% of the prenatal care visits in the beneficiary population), and randomly assigned 18 to treatment and 19 to control. To the best of our understanding, the treatment assignment was balanced across the 37 clinics and was not stratified. The intervention was implemented only for eight months (May 2010 - December 2010), and the clinics were clearly informed of the temporary nature of this intervention. During the intervention period, control group clinics saw no change in their fees for prenatal visits. In contrast, clinics in the treatment group received a three-fold increase in payments for any first prenatal visit that occurred before week 13 of pregnancy. Prenatal visits after week 13 or subsequent prenatal visits experienced no change in fees.

\begin{remark}\label{rem:app_sizes}
This application features a setting where all patients from the sampled clinics were included in their study. In terms of our notation, we have $|\mathcal{M}_{g}| = N_g$ for all $1 \leq g \leq G$. As a consequence, the weights in the weighted average of averages estimator $\hat{\theta}_{2,G}$ equal one, and the estimator coincides with the difference-in-means estimator, i.e., $\hat{\theta}_{2,G} = \hat{\theta}^{\rm alt}_{G}$.
\end{remark}

\cite{celhay/gertler/giovagnoli/vermeersch:2019} collected data before, during, and after the intervention period. The pre-intervention ran between January 2009 and April 2010. During this period, the sample average of the week of the first prenatal visit is 16.97, and only 34.46\% of these visits occur before week 13. Figure \ref{fig:histogram_week} shows a histogram of this distribution. The aforementioned treatment occurred exclusively during the intervention period, which ran between May 2010 and December 2010. While the treatment and control group affected approximately the same number of facilities, the number of treated and control patients are very different due to the unequal number of patients across facilities. Figure \ref{fig:histogram_clinic} provides a histogram of the number of patients attending each clinic for their first prenatal visit during the intervention period. This distribution has a mean and a standard deviation of 33.6 and 16.3 patients per clinic, respectively.  Finally, the post-intervention period goes between January 2011 and March 2012.

\begin{figure}
\begin{center}
\includegraphics[width=0.8\textwidth]{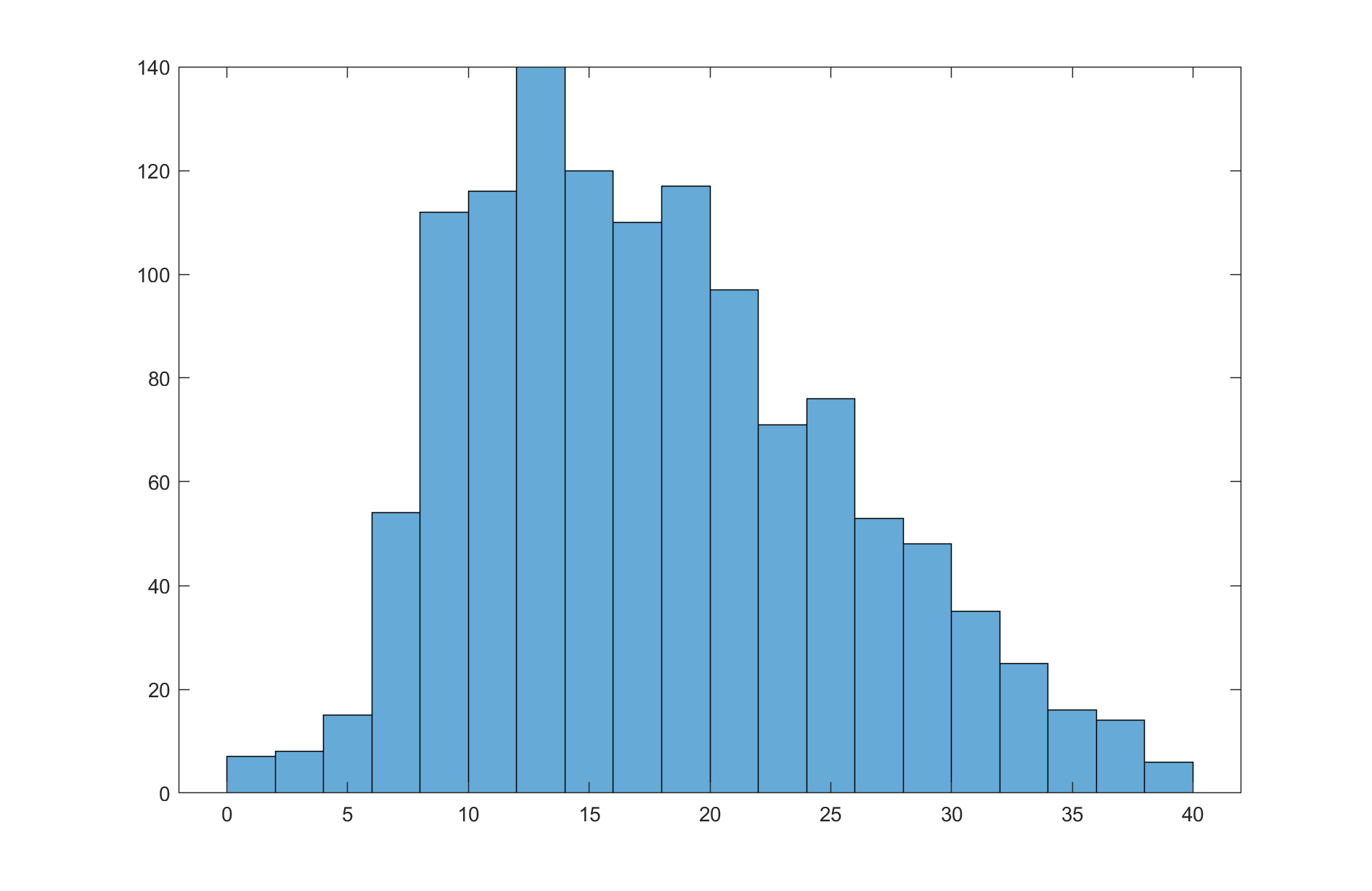}
\caption{Histogram of the week of the first prenatal visit during the pre-intervention period.}\label{fig:histogram_week}
\end{center}
\end{figure}

\begin{figure}
\begin{center}
\includegraphics[width=0.8\textwidth]{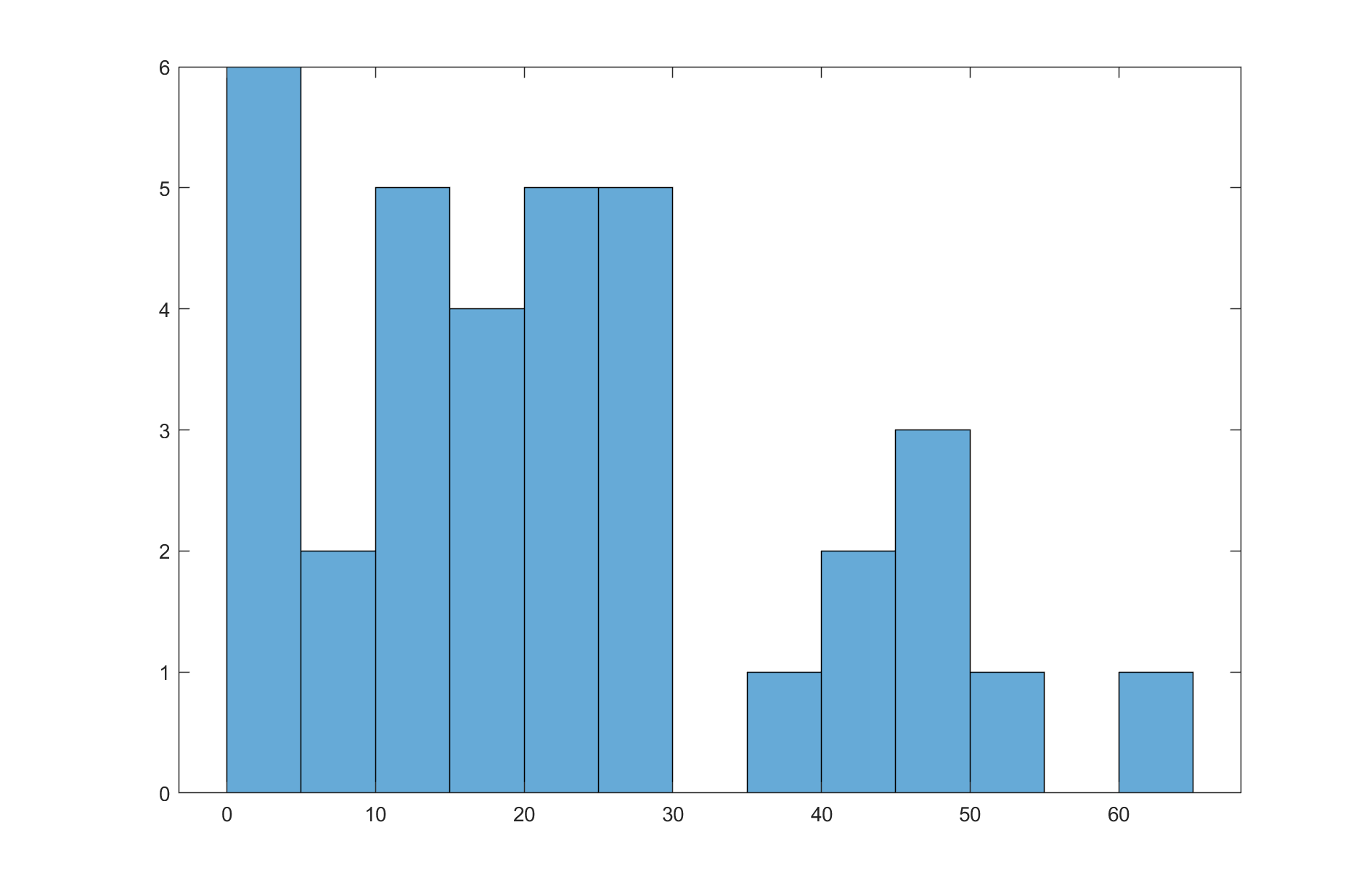}
\caption{Histogram of the patients per clinic having a first prenatal visit during the intervention period.}\label{fig:histogram_clinic}
\end{center}

\end{figure}

In Table \ref{tab:application}, we report our estimates of the equally-weighted and size-weighted average treatment effect (ATE) and their standard errors.\footnote{In the context of the experiment in \cite{celhay/gertler/giovagnoli/vermeersch:2019}, these would be more appropriately labeled as intention-to-treat (ITT) estimates, given the presence of imperfect compliance in the study.} We ran separate estimations using the data in the pre-intervention, intervention, and post-intervention periods. In each case, we considered two possible outcomes: $Y_1 =\text{weeks}$, which denotes the pregnancy week of the first prenatal visit, and $Y_2 =1\{\text{Weeks}<13\}$, which indicates whether the first prenatal visit occurs in pregnancy week 13 or lower. In all periods and for both outcomes, the equally-weighted cluster-level ATE, $\hat{\theta}_{1,G}$, appears to be small and statistically insignificant at the usual levels. In words, there seems to be no ATE at the clinic level. These results contrast sharply with those obtained using the size-weighted cluster-level ATE, $\hat{\theta}_{2,G}$, which measures ATE at the patient level. First, this ATE is not statistically significant during the pre-intervention period, which we consider a reasonable ``placebo-type'' finding. Second, during the intervention period, treated clinics had, on average, first prenatal visits 1.4 weeks earlier than control clinics. Moreover, the proportion of prenatal visits before week 13 was 10 percentage points higher in treated clinics than in control clinics during the same time period. These effects are statistically significant with $\alpha=5\%$ and seem economically important relative to the baseline levels. Interestingly, these effects seem to extend quantitatively to the post-intervention period, when the treatment incentives were completely removed. These findings demonstrate a statistically significant and economically important ATE at the patient level, both in the long and the short run. From the standpoint of our contribution, these differences between $\hat{\theta}_{1,G}$ and $\hat{\theta}_{2,G}$ point to the fact that cluster sizes are likely non-ignorable, as they may reflect essential factors such as the clinic's quality or the nearby population's size. To conclude, it is also worth noting that our results using the size-weighted cluster-level ATE align with the (intention-to-treat) estimates obtained by \cite{celhay/gertler/giovagnoli/vermeersch:2019}; this is expected given that, as explained in Remark \ref{rem:app_sizes}, $|\mathcal{M}_g| = N_g$ for all $1 \le g \le G$ in this application and so standard OLS regression recovers $\hat{\theta}_{2,G}$.

{\begin{table}[htbp]
  \centering
   \scalebox{0.97}{ \begin{tabular}{l|llll|llll|llll}
    \hline
\hline
     Period     & \multicolumn{4}{c|}{Pre-intervention} & \multicolumn{4}{c|}{Intervention} & \multicolumn{4}{c}{Post-intervention} \\
     Outcome     & \multicolumn{2}{c}{Weeks} & \multicolumn{2}{c|}{Weeks$<$13} & \multicolumn{2}{c}{Weeks} & \multicolumn{2}{c|}{Weeks$<$13} & \multicolumn{2}{c}{Weeks} & \multicolumn{2}{c}{Weeks$<$13} \\
     Estimator     & \multicolumn{1}{l}{$\hat{\theta}_{1,G}$} & \multicolumn{1}{l}{$\hat{\theta}_{2,G}$} & \multicolumn{1}{l}{$\hat{\theta}_{1,G}$} & \multicolumn{1}{l|}{$\hat{\theta}_{2,G}$} & \multicolumn{1}{l}{$\hat{\theta}_{1,G}$} & \multicolumn{1}{l}{$\hat{\theta}_{2,G}$} & \multicolumn{1}{l}{$\hat{\theta}_{1,G}$} & \multicolumn{1}{l|}{$\hat{\theta}_{2,G}$} & \multicolumn{1}{l}{$\hat{\theta}_{1,G}$} & \multicolumn{1}{l}{$\hat{\theta}_{2,G}$} & \multicolumn{1}{l}{$\hat{\theta}_{1,G}$} & \multicolumn{1}{l}{$\hat{\theta}_{2,G}$} \\\hline
    Estimate & 0.09  & -0.08 & 0.02  & 0.01  & -0.01 & -1.39** & 0.03  & 0.10***  & 0.11  & -1.59** & -0.02 & 0.09** \\
    s.e.  & 0.77  & 0.56  & 0.06  & 0.03  & 0.95  & 0.66  & 0.05  & 0.04  & 0.85  & 0.72  & 0.05  & 0.04 \\
\hline
\hline
\end{tabular}}
\caption{\small Estimation results based on data from \cite{celhay/gertler/giovagnoli/vermeersch:2019}. We ran our estimation separately for data in the pre-intervention period (i.e., Jan 2009 - Apr 2010), the intervention period (i.e., May 2010 - Dec 2010), and the post-intervention period (i.e., Jan 2011 - Mar 2012). The outcome variables are: ``Weeks'', which denotes the pregnancy week of the first prenatal visit, and ``Weeks$<$13'', which indicates whether the first prenatal visit occurred before pregnancy week 13. The significance level of the estimators is indicated with stars in the usual manner: ``***'' means significant at $\alpha=1\%$, ``**'' means significant at $\alpha=5\%$, and``*'' means significant at $\alpha=10\%$.}
\label{tab:application}
\end{table}}

\newpage
\appendix
\section{Appendix}
Throughout this appendix, we use LIE to denote the law of iterated expectations, LTP to denote the law of total probability, LLN to denote the Kolmogorov's strong law of large numbers, CLT to denote the Lindeberg-Levy central limit theorem, and CMT to denote the continuous mapping theorem.

\subsection{Auxiliary Results}

\begin{lemma} \label{lemma:iid}
Under Assumptions \ref{ass:QG}.(a)-(b), 
\[( (\bar{Y}_g(1), \bar{Y}_g(0), |\mathcal{M}_g|, Z_g, N_g ), ~1 \leq g \leq G)~,\]
is an i.i.d sequence of random variables.
\end{lemma}
\begin{proof}
Let $A_g := (\bar{Y}_g(1), \bar{Y}_g(0))$ and $B_g := (|\mathcal{M}_g|, Z_g, N_g)$. We first show that data are independent, i.e., for arbitrary vectors $a^{(G)}$ and $b^{(G)}$,
\begin{equation*}\label{eq:iid_state1}
P\{A^{(G)} \le a^{(G)}, B^{(G)} \le b^{(G)}\} ~=~ \prod_{1 \le g \le G}P\{A_g \le a_g, B_g \le b_g\}~,
\end{equation*}
where the inequalities are to be interpreted component-wise. To that end, let $C_g = (Y_{i,g}(1), Y_{i,g}(0): 1 \le i \le N_g)$ and denote by $\Gamma(a_g,N_g,\mathcal{M}_g)$ the (random) set such that $C_g \in \Gamma(a_g,N_g,\mathcal{M}_g)$ if and only if $A_g \le a_g$. Let $\Gamma^{(G)}$ denote the Cartesian product of $\Gamma(a_g,N_g,\mathcal{M}_g)$ for all $1 \le g \le G$. Then,
\begin{align*}
P\{A^{(G)} \le a^{(G)}, B^{(G)} \le b^{(G)}\} &~\overset{(1)}{=}~ P\{C^{(G)} \in \Gamma^{(G)}, B^{(G)} \le b^{(G)}\}\\
&~\overset{(2)}{=}~ E\left[E\left[I\{C^{(G)} \in \Gamma^{(G)}\}|\mathcal{M}^{(G)}, Z^{(G)}, N^{(G)}\right]I\{B^{(G)} \le  b^{(G)}\}\right]\\
&~\overset{(3)}{=}~ E\bigg[\prod_{1\le g \le G}E\left[C_g \in \Gamma(a_g,N_g,\mathcal{M}_g)|\mathcal{M}_g, Z_g, N_g\right]I\{B_g \le b_g\}\bigg] \\
&~\overset{(4)}{=}~ E\bigg[\prod_{1 \le g \le G}E\left[I\{C_g \in \Gamma(a_g,N_g,\mathcal{M}_g)\}I\{B_g \le b_g\}|\mathcal{M}_g, Z_g, N_g\right]\bigg]\\
&~\overset{(5)}{=}~ \prod_{1 \le g \le G}E\left[I\{C_g \in \Gamma(a_g,N_g,\mathcal{M}_g)\}I\{B_g \le b_g\}\right] \\ 
&~\overset{(6)}{=}~ \prod_{1 \le g \le G}P\{A_g \le a_g, B_g \le b_g\}~,
\end{align*}
where (1) and (6) follow from the definition of $\Gamma$, (2) from the LIE and the definition of $B_{g}$, (3) from Assumption \ref{ass:QG}.(b), (4) from the definition of $B_g$, and (5) from Assumption \ref{ass:QG}.(a) and the LIE.

Next, we show that data are identically distributed. To that end, consider the following derivation for arbitrary vectors $a$ and $b$, and $1 \le g, g' \le G$,
\begin{align*}
P\{A_g \le a, B_g \le b\} &~\overset{(1)}{=}~ P\{C_g \in \Gamma(a,N_g,\mathcal{M}_g), B_g \le b\}\\
&~\overset{(2)}{=}~ E\left[E\left[I\{C_g \in \Gamma(a,N_g,\mathcal{M}_g)\}|\mathcal{M}_g, Z_g, N_g\right]I\{B_g \le b\}\right]\\
&~\overset{(3)}{=}~ E\left[E\left[I\{C_{g'} \in \Gamma(a,N_{g'},S_{g'})\}|S_{g'}, Z_{g'}, N_{g'}\right]I\{B_{g'} \le b\}\right] \\
&~\overset{(4)}{=}~ P\{A_{g'} \le a, B_{g'} \le b\}~,
\end{align*}
where (1) follows from the definition of $\Gamma$, (2) and (4) from the LIE and the definition of $B_g$, and (3) from Assumptions \ref{ass:QG}.(a)--(b).
\end{proof}

\begin{proof}[Proof of Lemma \ref{prop:sample}]
First note that Assumption \ref{ass:QG}.(c) is satisfied by assumption.

We next show Assumption \ref{ass:QG}.(d). Fix $a\in\{0,1\}$ and $1\leq g\leq G$ arbitrarily. By the LIE,
\begin{equation}
    E[\bar{Y}_g(a)|N_g] ~=~ E\bigg[E\bigg[\frac{1}{|\mathcal{M}_g|}\sum_{i \in \mathcal{M}_g}Y_{i,g}(a)\bigg|Z_g, N_g, |\mathcal{M}_g|, (Y_{i,g}(a): 1 \le i \le N_g)\bigg]\bigg|N_g\bigg]~.
    \label{eq:sample1}
\end{equation}
The inner expectation can be viewed as the expectation of a sample mean of size $|\mathcal{M}_g|$ drawn from the set $(Y_{i,g}(a): 1 \le i \le N_g)$ without replacement. Hence, by \citet[Theorem 2.1]{cochran2007sampling},
\begin{equation}
    E\bigg[\frac{1}{|\mathcal{M}_g|}\sum_{i \in \mathcal{M}_g}Y_{i,g}(a)\bigg|Z_g, N_g, |\mathcal{M}_g|, (Y_{i,g}(a): 1 \le i \le N_g)\bigg] ~=~ \frac{1}{N_g}\sum_{1 \le i \le N_g}Y_{i,g}(a)~.
    \label{eq:sample2}
\end{equation}
The desired result follows from combining \eqref{eq:sample1} and \eqref{eq:sample2}.
\end{proof}

\subsection{Proof of Theorems}

\begin{proof}[Proof of Theorem \ref{theorem:limOLS}]
For arbitrary $a\in\{0,1\}$ and $1\leq g\leq G$, consider the following derivation.
\begin{align}
E\bigg[\bigg|\sum_{i\in \mathcal{M}_g}Y_{i,g}(a)\bigg|\bigg] 
&~\leq~ E\bigg[\sum_{i\in \mathcal{M}_g}|Y_{i,g}(a)|\bigg]\notag\\
&~\overset{(1)}{=}~ E\bigg[\sum_{i \in \mathcal{M}_g}E[|Y_{i,g}(a)||N_g, Z_g, \mathcal{M}_g]\bigg]\notag\\
&~\overset{(2)}{\leq}~ C^{1/2} E[|\mathcal{M}_g|] \le C^{1/2} E[N_g] ~\overset{(3)}{<}~  \infty~,\label{eq:limOLS}
\end{align}
where (1) follows from the LIE, (2) from Assumption \ref{ass:QG}.(f) and Jensen's inequality, and (3) from Assumption \ref{ass:QG}.(e) and Jensen's inequality. We further have by Assumption \ref{ass:QG}.(e) that $E[|\mathcal{M}_g|] < \infty$. 

Under Assumptions \ref{ass:assignment}--\ref{ass:QG}, Lemma \ref{lemma:iid}, \eqref{eq:limOLS}, and $E[|\mathcal{M}_g|] < \infty$, the desired result follows from Lemma C.4 in \cite{bugni/canay/shaikh:2019} and the CMT.
\end{proof}

\begin{proof}[Proof of Theorem \ref{theorem:mainECA}]
For any $a\in\{0,1\}$ and $1\leq g\leq G$, consider the following preliminary derivation.
\begin{align}
E[\bar{Y}_g(a)^2] 
~\overset{(1)}{\le}~ E\bigg[\frac{1}{|\mathcal{M}_g|}\sum_{i \in \mathcal{M}_g}Y_{i,g}(a)^2\bigg] 
~\overset{(2)}{=}~ E\bigg[\frac{1}{|\mathcal{M}_g|}\sum_{i \in \mathcal{M}_g}E[Y_{i,g}(a)^2|N_g, Z_g, \mathcal{M}_g]\bigg]~\overset{(3)}{\le} C ~<~ \infty~,\label{eq:3p2_0}
\end{align}
where (1) follows from Jensen's inequality, (2) from the LIE, and (3) from Assumptions \ref{ass:QG}.(f)-(c). Next, note that by Assumption \ref{ass:QG}.(d) and the LIE,
\begin{equation}
\theta_1 ~=~ E\bigg[\frac{1}{N_g}\sum_{1 \le i \le N_g}Y_{i,g}(1) - Y_{i,g}(0)\bigg] ~=~ E\left[\bar{Y}_g(1) - \bar{Y}_g(0)\right]~.\label{eq:3p2_1}
\end{equation}

Under \eqref{eq:3p2_0}, \eqref{eq:3p2_1}, and Assumptions \ref{ass:assignment}--\ref{ass:QG}, the desired result follows immediately from the proof of Theorem 4.1 in \cite{bugni/canay/shaikh:2018} where the clusters are viewed as the experimental units with potential outcomes given by $(\bar{Y}_g(0), \bar{Y}_g(1))$.
\end{proof}

\begin{proof}[Proof of Theorem \ref{thm:limHC}]
Fix $a \in\{0,1\}$ and  $r \in \{0, 1, 2\}$ arbitrarily. For any $1\leq g\leq G$, we can repeat arguments in the proof of Theorem \ref{theorem:mainECA} to show that
\begin{equation}
E[ \bar{Y}_{g}(a)^{r}] ~<~\infty ~.
\label{eq:thm_3p3_1}
\end{equation}
Under Assumptions \ref{ass:assignment}--\ref{ass:QG}, Lemma \ref{lemma:iid}, and \eqref{eq:thm_3p3_1}, Lemma C.4 in \cite{bugni/canay/shaikh:2019} implies that
\begin{align}
\frac{1}{G}\sum_{1\le g\le G}\bar{Y}_{g}^{r}I\{A_{g}=a\}&~\xrightarrow{p}~P\{A_{g} = a\} E[\bar{Y}_{g}(a)^{r}]
\label{eq:thm_3p3_2} ~.
\end{align}
From this and the CMT, we conclude that
\begin{align}
\widehat{\var}_{G}[\bar{Y}_{g}(a)]& ~{\xrightarrow{p}}~
\var[\bar{Y}_g(a)]~ .
\label{eq:thm_3p3_3}
\end{align}
To conclude the proof, we note that the convergence in \eqref{eq:limHC} follows from \eqref{eq:thm_3p3_2} (for $r=0$ and $a \in \{0,1\}$), \eqref{eq:thm_3p3_3} (for $a \in \{0,1\}$), $\pi \in (0,1)$, and the CMT. Also, the inequality in \eqref{eq:limHC} follows from the fact that $\var[\bar{Y}^{\dagger}_g(a)] = \var[\bar{Y}_g(a)] - \var[E[\bar{Y}_g(a)|S_g]]$ and $\var[E[\bar{Y}_g(a)|S_g]] = E[\bar{m}_a(S_g)^2]$ for $a \in \{0,1\}$. Some additional algebra will confirm the necessary and sufficient conditions for the inequality to become an equality.
\end{proof}

\begin{proof}[Proof of Theorem \ref{theorem:sigma1}]
The result follows immediately from the proof of Theorem 4.2 in \cite{bugni/canay/shaikh:2018} where the clusters are viewed as the experimental units with potential outcomes given by $(\bar{Y}_g(0), \bar{Y}_g(1))$.
\end{proof}

\begin{proof}[Proof of Theorem \ref{thm:mainSECT}]
We follow the general strategy in the proof of Theorem 4.1 in \cite{bugni/canay/shaikh:2018}, as extended in \cite{liu2023inference}.
By the LIE and Assumption \ref{ass:QG}.(d), $E[ \sum_{1 \leq i \leq N_g} Y_{i,g}(a)] = E[\bar{Y}_g(a) N_g]$ for $a \in \{0, 1\}$,
and thus
\begin{equation*}
    \theta_2 = \frac{E\left[\bar{Y}_g(1)N_g\right]}{E[N_g]} - \frac{E\left[\bar{Y}_g(0) N_g\right]}{E[N_g]} ~.
    \end{equation*}
As a consequence,
\begin{equation}
    \sqrt{G}(\hat{\theta}_{2,G} - \theta_2) ~=~ \sqrt{G}(h(\hat{\Theta}_G) - h(\Theta))~,
    \label{eq:mainSECT1}
\end{equation}
where the function $h: \mathbb{R}^4 \rightarrow \mathbb{R}$ defined as $h(w, x, y, z) := \frac{w}{x} - \frac{y}{z}$, $G_{1} :=\frac{1}{G}\sum_{1\le g\le G}I\{A_{g}=1\}$, and
\begin{align*}
\hat{\Theta}_G &~:=~ \begin{pmatrix}
\frac{1}{G_1}\sum_{1 \le g \le G}\bar{Y}_g(1)N_gI\{A_g = 1\}  \\
\frac{1}{G_1}\sum_{1 \le g \le G}N_gI\{A_g = 1\}  \\
\frac{1}{G - G_1}\sum_{1 \le g \le G}\bar{Y}_g(0)N_gI\{A_g = 0\}  \\
\frac{1}{G - G_1}\sum_{1 \le g \le G}N_gI\{A_g = 0\}
\end{pmatrix}~,\\
\Theta &~:=~ (E[\bar{Y}_g(1)N_g], E[N_g], E[\bar{Y}_g(0)N_g], E[N_g])'~.
\end{align*}

By \eqref{eq:mainSECT1}, we derive our result by characterizing the asymptotic distribution of $\sqrt{G}(\hat{\Theta}_G - \Theta)$ and applying the Delta method. To this end, note that
\begin{equation}
    \sqrt{G}\left( \hat{\Theta}_{G}-\Theta \right) ~\overset{(1)}{=}~\left( 
\begin{array}{cccc}
\frac{\pi }{G_{1}/G} & 0 & 0 & 0 \\ 
0 & \frac{\pi }{G_{1}/G} & 0 & 0 \\ 
0 & 0 & \frac{1-\pi }{1-G_{1}/G} & 0 \\ 
0 & 0 & 0 & \frac{1-\pi }{1-G_{1}/G}%
\end{array}%
\right) \sqrt{G}\mathbb L_G ~\overset{(2)}{=}~ (I + o_p(1)) \sqrt{G}\mathbb L_G ~,\label{eq:mainSECT2}
\end{equation}
where (1) uses that $\mathbb L_G := \left( L_G^{\rm YN1}, L_G^{\rm N1}, L_G^{\rm YN0}, L_G^{\rm N0}\right)'$ with
\begin{align*}
    L_G^{\rm YN1} ~&:=~ \frac{1}{\pi}\bigg(\frac{1}{G}\sum_{1 \leq g \leq G} \left( \bar{Y}_g(1) N_g - E\left[\bar{Y}_g(1) N_g\right]\right) I\{A_g= 1\}\bigg)~, \\
     L_G^{\rm N1}  ~&:=~ \frac{1}{\pi}\bigg(\frac{1}{G}\sum_{1 \leq g \leq G} \left( N_g -E\left[ N_g\right]\right) I\{A_g=1\}\bigg)~, \\
    L_G^{\rm YN0}  ~&:=~ \frac{1}{1 - \pi}\bigg(\frac{1}{G}\sum_{1 \leq g \leq G} \left( \bar{Y}_g(0) N_g -E\left[\bar{Y}_g(0) N_g\right]\right)I\{A_g=0\}\bigg)~, \\
     L_G^{\rm N0}  ~&:=~ \frac{1}{1 - \pi}\bigg(\frac{1}{G}\sum_{1 \leq g \leq G} \left( N_g - E\left[ N_g\right]\right) I\{A_g=0\}\bigg) ~,
\end{align*}
and (2) follows from Assumption \ref{ass:assignment}, as it implies that  ${G_1}/{G} 
    ~=~ \sum_{s \in S}{D_G(s)}/{G} + \pi~ {\xrightarrow{p}}~ \pi$. 

By \eqref{eq:mainSECT2}, the next step to characterize the asymptotic distribution of $\sqrt{G}(\hat{\Theta}_G - \Theta)$ is to find the asymptotic distribution of $\sqrt{G}\mathbb{L}_G$. To this end, for any $a \in \{0,1\}$, we define
\begin{align*}
    \mathbf{d} ~&:=~ \left({D_G(s)}/{\sqrt{G}}~:~ s\in \mathcal{S} \right)^\prime\\
    \mathbf{p} ~&:=~\left(\sqrt{G} \left({G(s)}/{G} - p(s) \right)~:~ s\in \mathcal{S} \right)^\prime\\
    \mathbf{m}^{\rm Y}_{a} ~&:=~ \left(E\left[\bar{Y}_g(a)N_g\mid S_g = s\right]-E\left[\bar{Y}_g(a)N_g\right] ~:~ s\in \mathcal{S}\right)^\prime\\
    \mathbf{m}^{\rm N}~&:=~ \left(E\left[N_g\mid S_g = s\right]-E\left[N_g\right] ~:~ s\in \mathcal{S}\right)^\prime~,
\end{align*}
where $G(s)$ is as in \eqref{eq:Gs}. By some algebra, we find that
\begin{align}
\scriptsize
\sqrt{G}\mathbb L_G ~=~   \underbrace{\begin{pmatrix}
1 & 0 & 0 & 0 &\frac{1}{\pi} \left(\mathbf{m}_{1}^{\rm Y}\right)^\prime & \left(\mathbf{m}_{1}^{\rm Y}\right)^\prime\\
0 & 1 & 0 & 0 &\frac{1}{\pi} \left(\mathbf{m}^{\rm N}\right)^\prime & \left(\mathbf{m}^{\rm N}\right)^\prime \\
0 & 0 & 1 & 0 &-\frac{1}{1-\pi} \left(\mathbf{m}_{0}^{\rm Y}\right)^\prime & \left(\mathbf{m}_{0}^{\rm Y}\right)^\prime \\
0 & 0 & 0 & 1 &-\frac{1}{1-\pi} \left(\mathbf{m}^{\rm N}\right)^\prime & \left(\mathbf{m}^{\rm N}\right)^\prime
\end{pmatrix}}_{=:~ B^\prime} \underbrace{\begin{pmatrix}
\frac{1}{\sqrt{G}} \sum_{g=1}^{G} \frac{1}{\pi}(\bar{Y}_g(1)N_g - E[\bar{Y}_g(1)N_g|S_g])  I\{A_g=1\} \\
\frac{1}{\sqrt{G}} \sum_{g=1}^{G} \frac{1}{\pi}(N_g - E[N_g|S_g])I\{A_g=1\} \\
\frac{1}{\sqrt{G}} \sum_{g=1}^{G} \frac{1}{1-\pi} (\bar{Y}_g(0)N_g - E[\bar{Y}_g(0)N_g|S_g]) I\{A_g=0\} \\
\frac{1}{\sqrt{G}} \sum_{g=1}^{G} \frac{1}{1-\pi}(N_g - E[N_g|S_g])I\{A_g=0\}\\
\mathbf{d} \\
\mathbf{p}
\end{pmatrix}}_{=: ~\mathbf{y}_G}~.
\label{eq:mainSECT3a}
\end{align}
Under Assumptions \ref{ass:assignment}--\ref{ass:QG}, one can follow the partial sum and decomposition arguments developed in Lemma B.2 of  \cite{bugni/canay/shaikh:2018} (or, equivalently, Lemma C.1 in \cite{bugni/canay/shaikh:2019}), to obtain
\begin{equation}
\mathbf{y}_G ~\stackrel{d}{\rightarrow}~ \mathcal{N}(0, \Sigma)~,
\label{eq:mainSECT3b}
\end{equation}
where
\begin{equation*}
    \Sigma ~=~ \begin{pmatrix}
    \Sigma_1 & 0 & 0 & 0\\
    0 & \Sigma_0 & 0 & 0 \\
    0 & 0 & \Sigma_D & 0 \\
    0 & 0 & 0 & \Sigma_{G} 
    \end{pmatrix}~,
\end{equation*}
and, for $a \in \{0,1\}$,
\begin{align*}
    &\Sigma_a := \frac{1}{P\{A_g = a\}}\begin{pmatrix}
    {E\left[\text{Var}\left[\bar{Y}_g(a)N_g | S_g\right]\right]} & {E\left[\text{Cov}\left[\bar{Y}_g(a)N_g, N_g| S_g\right]\right]} \\
    {E\left[\text{Cov}\left[\bar{Y}_g(a)N_g, N_g| S_g\right]\right]} & {E\left[\text{Var}\left[N_g\right|S_g]\right]}
    \end{pmatrix}~, \\
    &\Sigma_D := \text{diag}\left(p(s)\tau(s): s\in\mathcal{S} \right)~, \\ 
    &\Sigma_G := \text{diag}\left(p(s): s\in\mathcal{S} \right)- \left(p(s): s\in\mathcal{S} \right) \left(p(s): s\in\mathcal{S} \right)^\prime ~.
\end{align*}

By combining \eqref{eq:mainSECT2}, \eqref{eq:mainSECT3a}, and \eqref{eq:mainSECT3b}, we conclude that
\begin{equation}
    \sqrt{G}(\hat{\Theta}_G - \Theta) ~=~  \sqrt{G}\mathbb{L}_G + o_P(1) ~\stackrel{d}{\rightarrow}~ \mathcal{N}(0, \mathbb V)~,
    \label{eq:mainSECT4}
\end{equation}
where $\mathbb V := B^\prime \Sigma B$. By additional calculations involving the law of total variance/covariance, $\mathbb V$ is a 4x4 symmetric matrix whose components are given by
\begin{align*}
\mathbb V_{11} & ~=~  \frac{1}{\pi} \var \left[\bar{Y}_g(1) N_g\right] -  \frac{1-\pi}{\pi} \var \left[E\left[\bar{Y}_g(1) N_g \mid S_g \right]\right]+ E\left[\frac{\tau(S_g)}{\pi^2} \left( E[\bar{Y}_g(1)N_g\mid S_g] - E[\bar{Y}_g(1)N_g] \right)^2 \right] \\
\mathbb V_{12} & ~=~ \frac{1}{\pi}\cov[\bar{Y}_g(1) N_g, N_g] - \frac{1-\pi}{\pi}\cov[E[\bar{Y}_g(1) N_g | S_g], E[N_g | S_g]]\\
& \quad + E\left[\frac{\tau(S_g)}{\pi^2} \left( E[\bar{Y}_g(1)N_g\mid S_g] - E[\bar{Y}_g(1)N_g] \right) \left( E[N_g\mid S_g] - E[N_g] \right)  \right] \\
\mathbb V_{13} & ~=~  \cov[E[\bar{Y}_g(1) N_g | S_g], E[\bar{Y}_g(0) N_g | S_g]] \\
& \quad - E\left[\frac{\tau(S_g)}{\pi(1-\pi)} \left( E[\bar{Y}_g(1)N_g\mid S_g] - E[\bar{Y}_g(1)N_g] \right) \left( E[\bar{Y}_g(0)N_g\mid S_g] - E[\bar{Y}_g(0)N_g] \right)  \right]\\
\mathbb V_{14} & ~=~  \cov[E[\bar{Y}_g(1) N_g | S_g], E[N_g | S_g]]- E\left[\frac{\tau(S_g)}{\pi(1-\pi)} \left( E[\bar{Y}_g(1)N_g\mid S_g] - E[\bar{Y}_g(1)N_g] \right) \left( E[N_g\mid S_g] - E[N_g] \right)  \right]\\
\mathbb V_{22} & ~=~ \frac{1}{\pi}\var[N_g] - \frac{1-\pi}{\pi} \var[E[N_g | S_g]] + E\left[\frac{\tau(S_g)}{\pi^2} \left( E[N_g\mid S_g] - E[N_g] \right)^2  \right]\\
\mathbb V_{23} &~=~  \cov[E[N_g | S_g], E[\bar{Y}_g(0) N_g | S_g]]  - E\left[\frac{\tau(S_g)}{\pi(1-\pi)} \left( E[N_g\mid S_g] - E[N_g] \right) \left( E[\bar{Y}_g(0)N_g\mid S_g] - E[\bar{Y}_g(0)N_g] \right)  \right]\\
\mathbb V_{24} & =  \var[E[N_g | S_g]] 
- E\left[\frac{\tau(S_g)}{\pi(1-\pi)} \left( E[N_g\mid S_g] - E[N_g] \right)^2  \right]\\
\mathbb V_{33} & ~=~ \frac{1}{1-\pi} \var[\bar{Y}_g(0) N_g] - \frac{\pi}{1-\pi} \var[E[\bar{Y}_g(0) N_g | S_g]]+ E\left[\frac{\tau(S_g)}{(1-\pi)^2}\left( E[\bar{Y}_g(0)N_g\mid S_g] - E[\bar{Y}_g(0)N_g] \right)^2  \right]\\
\mathbb V_{34} & ~=~ \frac{1}{1-\pi} \cov[\bar{Y}_g(0) N_g, N_g] - \frac{\pi}{1-\pi} \cov[E[\bar{Y}_g(0) N_g | S_g], E[N_g | S_g]] \\
& \quad + E\left[\frac{\tau(S_g)}{(1-\pi)^2}  \left( E[\bar{Y}_g(0)N_g\mid S_g] - E[\bar{Y}_g(0)N_g] \right) \left( E[N_g\mid S_g] - E[N_g] \right)  \right]\\
\mathbb V_{44} & ~=~ \frac{1}{1-\pi} \var[N_g] - \frac{\pi}{1-\pi}\var[E[N_g | S_g]]
+ E\left[\frac{\tau(S_g)}{(1-\pi)^2} \left( E[N_g\mid S_g] - E[N_g] \right)^2  \right] ~.
\end{align*}

By \eqref{eq:mainSECT1}, \eqref{eq:mainSECT4}, and the Delta method, we get
\begin{equation*}
    \sqrt{G}(\hat{\theta}_2 - \theta_2) ~=~ \sqrt{G}(h(\hat{\Theta}) - h(\Theta)) ~\stackrel{d}{\rightarrow} ~\mathcal{N}\big(~0, ~(\nabla h_0)^\prime \mathbb V(\nabla h_0)~\big) ~,
\end{equation*}
where
\begin{equation*}
    \nabla h_0 ~:=~ \left ( \frac{1}{E[N_g]}, - \frac{E[\bar{Y}_g(1)N_g]}{E[N_g]^2}, - \frac{1}{E[N_g]}, \frac{E[\bar{Y}_g(0)N_g]}{E[N_g]^2} \right )^\prime ~.
\end{equation*}

To conclude the proof, we need to show that $\sigma^2_2$ in \eqref{eq:mainSECT0} is equal to $(\nabla h_0)^\prime \mathbb V(\nabla h_0)$. This follows from additional algebraic calculations similar to those in the proof of Theorem 3.1 in \cite{bai2022pairs}.
\end{proof}

\begin{proof}[Proof of Theorem \ref{thm:limCR}]
First, we verify that $\tilde{\sigma}^2_2$ can be written as in \eqref{eq:hatsigma2exp} with \eqref{eq:hatsigma2}. To that end, let ${\bf 1}_{K}$ denote a column vector of ones of length $K$. The cluster-robust variance estimator can then be written as
\begin{align*}
G \Bigg(\sum_{1\le g \le G}X_g'X_g\Bigg)^{-1} \Bigg( \sum_{1 \le g \le G} X_g' \hat{\epsilon}_g\hat{\epsilon}_g' X_g \Bigg) \Bigg(\sum_{1 \le g \le G}X_g'X_g\Bigg)^{-1},
\end{align*}
where 
\begin{align*}
X_g~&:=~\left(\begin{array}{cc}
{\bf 1}_{|\mathcal{M}_g|} \cdot \sqrt{\frac{N_g}{|\mathcal{M}_g|}} & ~~~~{\bf 1}_{|\mathcal{M}_g|} \cdot \sqrt{\frac{N_g}{|\mathcal{M}_g|}} A_g
\end{array}%
\right)\\
\hat{\epsilon}_g ~&:=~ \left(\hat{\epsilon}_{i,g}(1)\sqrt{\tfrac{N_g}{|\mathcal{M}_g|}}A_g + \hat{\epsilon}_{i,g}(0)\sqrt{\tfrac{N_g}{|\mathcal{M}_g|}}(1 - A_g)~:~ i \in \mathcal{M}_g\right)'~.
\end{align*}
By doing some algebra, it follows that
\[\sum_{1 \le g \le G} X_g'X_g = \begin{pmatrix}
\sum_{1 \le g \le G} N_g & \sum_{1 \le g \le G} N_g A_g \\
\sum_{1 \le g \le G} N_g A_g & \sum_{1 \le g \le G} N_g A_g
\end{pmatrix},\]
and
\begin{align*}
    &\sum_{1 \le g \le G}X_g' \hat{\epsilon}_g\hat{\epsilon}_g' X_g =\\
    & \sum_{1 \le g \le G} A_g \left(\frac{N_g}{|\mathcal{M}_g|}\right)^2
 \Big(\sum_{i \in \mathcal{M}_g}\hat{\epsilon}_{i,g}(1) \Big)^2
\begin{pmatrix}
1 & 1\\
1 & 1
\end{pmatrix}
+
 \sum_{1 \le g \le G} (1-A_g)
 \left(\frac{N_g}{|\mathcal{M}_g|}\right)^2\Big(\sum_{i \in \mathcal{M}_g}\hat{\epsilon}_{i,g}(0)\Big)^2
\begin{pmatrix}
1 & 0 \\
0 & 0
\end{pmatrix}.
\end{align*}
The desired result then follows from further algebraic manipulations based on these expressions. 

Next, we prove \eqref{eq:thm3p6}. To this end, fix $a \in\{0,1\}$,  $r \in \{0, 1, 2\}$, and $l \in \{1,2\}$ arbitrarily. 
For any $1\leq g\leq G$, we can repeat arguments in the proof of Theorem \ref{theorem:mainECA} to show that
\begin{equation}
E[ N_{g}^{l}\bar{Y}_{g}( a) ^{r}] ~<~\infty ~. \label{eq:thm_3p5_1}
\end{equation}
Under Assumptions \ref{ass:assignment}--\ref{ass:QG}, Lemma \ref{lemma:iid}, and \eqref{eq:thm_3p5_1}, Lemma C.4. in \cite{bugni/canay/shaikh:2019} implies that
\begin{align}
\frac{1}{G}\sum_{1\leq g\leq G}N_{g}^{l}\bar{Y}_{g}^{r}( a) I\{A_{g}=a\}& ~\xrightarrow{p}~E[N_{g}^{l}\bar{Y}_{g}(a)^{r}]P\{A_{g}=a\} ~.\label{eq:thm_3p5_2}
\end{align}
From these results, consider the following derivation.
\begin{align}
\tilde{\sigma}_{2,G}^{2}( a) &~=~\frac{\frac{1}{G}\sum_{1\leq g\leq G}N_{g}^{2}I\{A_{g}=a\}( \frac{1}{|S_{g}|}\sum_{i\in S_{g}}\hat{\epsilon }_{i,g}(a)) ^{2}}{( \frac{1}{G}\sum_{1\leq g\leq G}{N_{g}} I\{A_{g}=a\}) ^{2}} \notag \\
&~\overset{(1)}{=}~
~\frac{\frac{1}{G}\sum_{1\leq g\leq G}N_{g}^{2}\bar{Y}_{g}( a) ^{2}I\{A_{g}=a\}}{( \frac{1}{G}\sum_{1\leq  g\leq G}{N_{g}} I\{A_{g}=a\}) ^{2}}~+~\frac{( \frac{1}{G}\sum_{1\leq g\leq G}N_{g}^{2}I\{A_{g}=a\}) ( \frac{1}{G}\sum_{1\leq g\leq G}N_{g} \bar{Y}_{g}( a) I\{A_{g}=a\}) ^{2}}{( \frac{1}{G} \sum_{1\leq g\leq G}{N_{g}}I\{A_{g}=a\}) ^{4}} \notag\\
&\quad~~~-~2\frac{( \frac{1}{G}\sum_{1\leq g\leq G}N_{g}^{2}I\{A_{g}=a\}\bar{Y} _{g}( a) ) \frac{1}{G}\sum_{1\leq g\leq G}N_{g}\bar{Y} _{g}( a) I\{A_{g}=a\}}{( \frac{1}{G}\sum_{1\leq g\leq G}{ N_{g}}I\{A_{g}=a\}) ^{3}}
\notag \\
&~\overset{(2)}{\xrightarrow{p}}~\frac{E[ N_{g}^{2}\bar{Y}_{g}( a) ^{2}]+\frac{E[ N_{g}^{2}] ( E[ N_{g}\bar{Y} _{g}( a) ] ) ^{2}}{( E[ N_{g}] ) ^{2}}  -2\frac{E[ N_{g}^{2}\bar{Y}_{g}( a) ] E[ N_{g}\bar{Y}_{g}( a) ] }{E[ N_{g} ] }}{E[ {N_{g}}] ^{2}P\{A_{g}=a\}} ~\overset{(3)}{=}~\frac{1}{P\{A_{g}=a\}}\var[\tilde{Y}_g(a)] ~,\label{eq:thm_3p5_3}
\end{align}
where (1) follows from the definition of $\hat{\epsilon}_{i,g}(a) $, (2) from \eqref{eq:thm_3p5_2} (with $r \in \{0, 1, 2\}$, and $l \in \{1,2\}$), $E[ N_{g}]>0$, $ P\{A_{g}=a\}>0$, and the CMT, and (3) from the definition of $\tilde{Y}_g(a)$. 

To conclude the proof, we note that the convergence in \eqref{eq:thm3p6} follows from \eqref{eq:thm_3p5_3} (for $a \in \{0,1\}$) and the CMT. Also, the inequality in \eqref{eq:thm3p6} follows from the fact that $\var[\tilde{Y}^{\dagger}_g(a)] = \var[\tilde{Y}_g(a)] - \var[E[\tilde{Y}_g(a)|S_g]]$ and $\var[E[\tilde{Y}_g(a)|S_g]] = E[\tilde{m}_a(S_g)^2]$ for $a \in \{0,1\}$. Some additional algebra will confirm the necessary and sufficient conditions for the inequality to become an equality.
\end{proof}

\begin{proof}[Proof of Theorem \ref{theorem:sigma2}]
The result follows by studying the probability limit of each component and then applying the CMT. Specifically, it will follow once we establish that, for $a \in \{0, 1\}$,
\begin{align}
\hat{\mu}^{\hat{Y}}_{G,a}(s) &~\stackrel{p}{\rightarrow}~ E[\tilde{Y}_g(a)|S_g = s]~,\label{eq:hat_consist_1}\\
\hat{\mu}^{\hat{Y}^2}_{G,a} &~\stackrel{p}{\rightarrow}~ E[\tilde{Y}^2_g(a)]~.\label{eq:hat_consist_2}
\end{align} 

We only show \eqref{eq:hat_consist_1}, as \eqref{eq:hat_consist_2} can be shown by similar arguments. To this end, fix $a \in \{0,1\}$ and $r \in \{0,1\}$ arbitrarily. Under Assumptions \ref{ass:assignment}--\ref{ass:QG}, Lemma C.4 in \cite{bugni/canay/shaikh:2019} implies that
\begin{equation}
    \frac{1}{G}\sum_{1 \le g \le G}\tilde{Y}_g(a)^{r} I\{A_g = a, S_g = s\} ~\stackrel{p}{\rightarrow}~ p(s) P\{A_g = a\} E[\tilde{Y}_g(a)^{r} |S_g = s]~.
    \label{eq:hat_consist_3}
\end{equation}
In turn, \eqref{eq:hat_consist_3} (for $r \in \{0,1\}$), $p(s) P\{A_g = a\}>0$, and the CMT imply that
\begin{equation}
    \hat{\mu}^{\widetilde{Y}}_{G,a}(s)~:=~ \frac{\sum_{1 \le g \le G}\tilde{Y}_g(a)I\{A_g = a, S_g = s\}}{\sum_{1 \le g \le G}I\{A_g = a, S_g = s\}} ~\stackrel{p}{\rightarrow}~ E[\tilde{Y}_g(a)|S_g = s]~.
    \label{eq:hat_consist_4}
\end{equation}

Next, consider the following derivation.
\begin{align}
&\frac{1}{G}\sum_{1 \le g \le G}(\hat{Y}_g - \tilde{Y}_g(a))I\{A_g = a, S_g = s\} \notag\\
&\overset{(1)}{=} 
 ~\frac{1}{\frac{1}{G}\sum_{1 \le j \le G}N_j}\frac{1}{G}\sum_{1 \le g \le G}N_g\bar{Y}_g(a)I\{A_g = a, S_g = s\} - \frac{1}{E[N_g]}\frac{1}{G}\sum_{1 \le g \le G}N_g\bar{Y}_g(a)I\{A_g = a, S_g = s\} \notag\\
&\quad~-~  \frac{1}{\frac{1}{G}\sum_{1 \le j \le G}N_j}\left(\frac{\frac{1}{G}\sum_{1 \le j \le G}\bar{Y}_j(a)I\{A_j = a\}N_j}{\frac{1}{G}\sum_{1 \le j \le G}I\{A_j = a\}N_j}\right)\frac{1}{G}\sum_{1 \le g \le G}N_gI\{A_g = a, S_g = s\} \notag\\
&\quad~+~ \frac{1}{E[N_g]}\frac{E[N_g\bar{Y}_g(a)]}{E[N_g]}\frac{1}{G}\sum_{1 \le g \le G}N_gI\{A_g = a, S_g = s\}~\overset{(2)}{\stackrel{p}{\rightarrow}} 0~,\label{eq:hat_consist_5}
\end{align}
where (1) follows from the definitions of $\hat{Y}_g$ and $\tilde{Y}_g(a)$, and (2) from repeated applications of the LLN, Lemma C.4 in \cite{bugni/canay/shaikh:2019}, and the CMT.

The desired result holds by the following derivation.
\begin{equation*}
    \hat{\mu}^{\hat{Y}}_{G,a}(s) ~\overset{(1)}{=}~ \hat{\mu}^{\widetilde{Y}}_{G,a}(s) + o_P(1) ~\overset{(2)}{\stackrel{p}\rightarrow}~ E[\tilde{Y}_g(a)|S_g = s]~,
\end{equation*}
where (1) follows from \eqref{eq:hat_consist_3} (for $r=0$), $p(s) P\{A_g = a\}>0$, \eqref{eq:hat_consist_5}, and the CMT, and (2) from \eqref{eq:hat_consist_4}.
\end{proof}

\subsection{Survey of Articles Published in AEJ:Applied 2018-2022}\label{sec:lit_table}
In this section, we document some relevant features of the sampling design and subsequent analyses from every article involving a cluster randomized experiment published in the AEJ: Applied from 2018 to 2022. Table \ref{table:lit_summary} summarizes our findings.\footnote{Although not formally documented, we also find that 10/15 of these papers explicitly mention the use of stratification in their experimental designs.} For each article, we record:
\begin{enumerate}
\item The sampling design used within clusters. We find that $6/15$ papers employ a sampling design such that $|\mathcal{M}_g| = \lfloor \gamma N_g \rfloor$ (all of these use $\gamma = 1$), one paper employs a sampling design such that $|\mathcal{M}_g| = k$ for $k$ fixed, and one paper features both of these designs. $7/15$ papers employ a sampling design that does not correspond to either of these designs (denoted ``other").
\item Whether the analyses in the paper discuss the use of sampling weights in their regressions and/or the effect that this may have on the parameter of interest. We find that $12/15$ papers do not feature such a discussion (however, in some cases it could be argued that an explicit discussion is not required).
\end{enumerate}

\begin{table}
\centering
\begin{tabular}{ccc}
  \hline

  Article
  & Sampling Design
  & \begin{tabular}{@{}c@{}}Discussion of Weighting?\end{tabular} \\
  \hline \hline

  \citet{gine2018together}& other & no \\

  \citet{lafortune2018role} & $|\mathcal{M}_g| = N_g$ & no \\

  \citet{mcintosh2018neighborhood} & other & yes \\
  \citet{bjorkman2019reducing} & other & no \\
  \citet{bolhaar2019job} & $|\mathcal{M}_g| = N_g$ & no \\
    \citet{busso2019causal} & other & no \\

  \citet{celhay/gertler/giovagnoli/vermeersch:2019} & $|\mathcal{M}_g| = N_g$ & no \\

  \citet{deserranno2019leader}
  & \begin{tabular}{@{}c@{}} other \\  \end{tabular}
  & yes \\

  \citet{loyalka2019does}& other & no \\

  \citet{bandiera2020women} & $|\mathcal{M}_g| = k$ & no \\

  \citet{banerjee2020governance}
  & \begin{tabular}{@{}c@{}} $|\mathcal{M}_g| = N_g$ or $|\mathcal{M}_g| = k$ \\ (depending on dep. variable) \end{tabular}
  & no \\
  \citet{mckenzie2021growing} & $|\mathcal{M}_g| = N_g$ & no \\
  \citet{mohanan2021different} & other & no \\

  \citet{muralidharan2021improving} & $|\mathcal{M}_g| = N_g$ & yes \\

  \citet{wheeler2022linkedin}
  & \begin{tabular}{@{}c@{}} $|\mathcal{M}_g| = N_g$ \\  \end{tabular}
  & no\\
   \hline \hline
\end{tabular}
\caption{\label{table:lit_summary} \small Summary of sampling designs and subsequent analyses for papers in AEJ:applied from 2018 to 2022. The second column describes the sampling scheme within cluster, where ``other'' indicates a sampling scheme other than $|\mathcal{M}_g| = \lfloor \gamma N_g \rfloor $ for some constant $\gamma$ or $|\mathcal{M}_g|=k$ for some constant $k$. The third column indicates whether the paper includes a discussion about the use of sampling weights and/or the effect that this may have on the parameter of interest.}
\end{table}

\subsection{Covariate Adjustment}\label{sec:adjust}

In this section, we present methods for linear covariate adjustment which exploit additional information in $Z_g$ and $N_g$ beyond that contained by $S_g$. In what follows, for stratum $s \in \mathcal{S}$, let $\Psi_{g,s} = \Psi_s(X_g)$ be a user-specified function of $X_g = (Z_g',N_g)'$. We note that $X_g$ could in principle include cluster-level aggregates of individual-level covariates, including intracluster means and quantiles, but we do not consider specifications that allow for individual-level covariates directly. In order to describe the adjustment procedures for $\theta_1$ and $\theta_2$ simultaneously, define the variable $V_g$ to be $\bar{Y}_g$ when estimating $\theta_1$ and to be $N_g\bar{Y}_g$ when estimating $\theta_2$, and define the variable $\nu_g$ to be $1$ when estimating $\theta_1$ and to be $N_g$ when estimating $\theta_2$. Let $I_a(s) := \{g: A_g = a, S_g = s\}$, $I(s) := \{g : S_g = s\}$. The estimators we consider are then given by
\[\hat{\theta}^{\rm adj}_{1,G} = \frac{1}{G}\sum_{1 \le g \le G}\hat{\Xi}_g~,\]
\[\hat{\theta}^{\rm adj}_{2,G} = \frac{1}{\sum_{1 \le g \le G}N_g}\sum_{1 \le g \le G}\hat{\Xi}_g~,\]
where
\[\hat{\Xi}_g = \frac{A_g(V_g - \hat{\eta}_1(S_g,X_g))}{\hat{\pi}(S_g)} - \frac{(1 - A_g)(V_g - \hat{\eta}_0(S_g, X_g))}{1 - \hat{\pi}(S_g)} + \hat{\eta}_1(S_g,X_g) - \hat{\eta}_0(S_g, X_g)\]
with $\hat{\pi}(s) = \frac{|I_a(s)|}{|I(s)|}$, and
\[\hat{\eta}_a(s, X_g) = \Psi_{g,s}'\hat{\beta}_{a,s}~,\]
where $\hat{\beta}_{a,s}$ are the regression coefficients obtained from running a regression on the observations in $I_a(s)$ of $V_g$ on a constant and $\Psi_{g,s}$.
Next, we define the corresponding variance estimators. Define
\[\hat{\sigma}^2_{1,\rm adj, G} = \frac{1}{G}\sum_{1 \le g \le G}\left[A_g\hat{\Omega}_1^2(S_g, V_g, X_g) + (1 - A_g)\hat{\Omega}_0^2(S_g, V_g, X_g) + \hat{\Omega}^2_2(S_g, V_g, X_g)\right]~,\]
\[\hat{\sigma}^2_{2,\rm adj, G} = \frac{\frac{1}{G}\sum_{1 \le g \le G}\left[A_g\hat{\Omega}_1^2(S_g, V_g, X_g) + (1 - A_g)\hat{\Omega}_0^2(S_g, V_g, X_g) + \hat{\Omega}^2_2(S_g, V_g, X_g)\right]}{\left(\frac{1}{G}\sum_{1 \le g \le G}N_g\right)^2}~,\]
where for $a \in \{0, 1\}$,
\[\hat{\Omega}_a(s,V_g,X_g) = \tilde{\Omega}_a(s, V_g, X_g) - \frac{1}{|I_a(s)|}\sum_{g \in I_a(s)}\tilde{\Omega}_a(s, V_g, X_g) - \hat{\theta}^{\rm adj}_{G}\left(\nu_g - \frac{1}{|I(s)|}\sum_{g \in I(s)}\nu_g\right)~,\]
with
\[\tilde{\Omega}_1(s,V_g,X_g) = \left(1 - \frac{1}{\hat{\pi}(s)}\right)\hat{\eta}_1(s,X_g) - \hat{\eta}_0(s,X_g) + \frac{V_g}{\hat{\pi}(s)} ~,\]
\[\tilde{\Omega}_0(s,V_g,X_g) = \left(\frac{1}{1 - \hat{\pi}(s)} - 1\right)\hat{\eta}_0(s,X_g) - \hat{\eta}_1(s,X_g) + \frac{V_g}{1 - \hat{\pi}(s)} ~,\]
and 
\[\hat{\Omega}_2(s, V_g, X_g) = \left(\frac{1}{|I_1(s)|}\sum_{g \in I_1(s)}V_g\right) -\left(\frac{1}{|I_0(s)|}\sum_{g \in I_0(s)}V_g\right) - \hat{\theta}^{\rm adj}_{G}\left(\frac{1}{|I(s)|}\sum_{g \in I(s)}\nu_g\right)~,\]
(where it is implicit that $\hat{\theta}^{\rm adj}_G$ denotes  the appropriate corresponding estimator in each case). Following the proof techniques introduced in this paper, as well as \cite{jiang2022improving} and \cite{wang2024model}, we should obtain under appropriate assumptions that 
\[\frac{\sqrt{G}(\hat{\theta}^{\rm adj}_{j,G} - \theta_j)}{\hat{\sigma}_{j,\rm adj, G}} \xrightarrow{d} N(0 ,1)~,\]
as $G \rightarrow \infty$, for $j \in \{1, 2\}$. Moreover, we conjecture that these estimators are the optimal \emph{linear}
adjustments for estimators of this form (that is, estimators which maintain linear working models for $\hat{\eta}_a(s,x)$). Finally, we conjecture that replacing $\hat{\eta}_a(s,x)$ by appropriate \emph{non-parametric} estimators would attain the efficiency bound derived in \cite{bai2023efficiency}.

\newpage
\bibliography{bibliography}

\end{document}